\documentclass[
final
]{dmtcs-episciences}
\usepackage[utf8]{inputenc}
\usepackage{amssymb} 
\usepackage{amsmath}
\usepackage{dsfont}
\usepackage[disable]{todonotes}
\usepackage[caption=false]{subfig}
\usepackage{amsthm}

\newtheorem{theorem}{Theorem}
\newtheorem{lemma}{Lemma}
\newtheorem{proposition}{Proposition}
\newtheorem{corollary}{Corollary}

\newtheorem{definition}{Definition}
\theoremstyle{definition}
\newtheorem*{claimproof}{\normalfont{\textit{Proof}}}

\usepackage{tikz}
\usetikzlibrary{arrows,shapes}
\tikzstyle{vertex}=[draw,circle,inner sep=1pt, minimum size=7pt]
\tikzstyle{bvertex}=[draw,fill=blue,circle,inner sep=1pt, minimum size=7pt]
\tikzstyle{black vertex}=[draw,fill=black,circle,inner sep=1pt, minimum size=7pt]
\tikzstyle{rvertex}=[draw,fill=red,circle,inner sep=1pt, minimum size=7pt]
\tikzstyle{yvertex}=[draw,fill={rgb:orange,50;yellow,20},circle,inner sep=1pt, minimum size=7pt]
\tikzstyle{svertex}=[draw,circle,inner sep=0pt, minimum size=6pt]
\tikzstyle{label}=[above=0.2cm]
\tikzstyle{label below}=[above=-1cm]
\tikzstyle{label right}=[shift={(1.5cm,0cm)}]
\tikzstyle{label sright}=[shift={(1.2cm,0cm)}]
\tikzstyle{blue dots} = [draw,dotted,line width=1pt,-,blue!100]
\tikzstyle{black dots} = [draw,dotted,line width=1pt,-,black!100]
\tikzstyle{red dots} = [draw,dotted, line width=1pt,-,red!80]
\tikzstyle{red edge} = [draw,line width=1.5pt,-,red!70]
\tikzstyle{lred edge} = [draw,line width=1.5pt,-,red!50]
\tikzstyle{blue edge} = [draw,line width=1.5pt,-,blue!100]
\tikzstyle{yellow edge} = [draw,line width=1.5pt,-,color={rgb:orange,30;yellow,20}]
\tikzstyle{black edge} = [draw,line width=1.5pt,-,black!100]
\tikzstyle{green edge} = [draw,line width=1.5pt,-,color={rgb:green,50;blue,20}]

\pgfdeclarelayer{background}
\pgfdeclarelayer{bg2}
\pgfsetlayers{bg2,background,main}

\newcommand{\todog}[1]{\todo[color=yellow!90!red]{ #1}}

\newcommand{\Oh}{\ensuremath{\mathcal{O}}}

\newtheorem{Claim}{Claim}

\renewcommand{\mod}{\text{ mod}}

\author[Eckstein et al.]{Nils Jakob Eckstein\affiliationmark{1}
  \and Niels Grüttemeier\affiliationmark{2}\thanks{The research was mainly performed when NG was associated with Philipps-Universität Marburg, Germany.} \\
  \and Christian Komusiewicz\affiliationmark{1}
  \and Frank Sommer\affiliationmark{1}\thanks{FS was supported by the DFG, project MAGZ (KO~3669/4-1).}}
\title{Destroying Multicolored Paths and Cycles\\in Edge-Colored~Graphs}
\affiliation{
  Fachbereich Mathematik und Informatik, Philipps-Universität Marburg, Germany\\
  Fraunhofer IOSB, Lemgo, Fraunhofer Institute of Optronics, System Technologies and Image Exploitation, Germany}

\begin{document}
\publicationdata{vol. 25:1}{2023}{7}{10.46298/dmtcs.7636}{2021-06-30; 2021-06-30; 2023-01-31}{2023-02-17}

\maketitle

\begin{abstract}
We study the computational complexity of \textsc{$c$-Colored $P_\ell$ Deletion} and \textsc{$c$-Colored $C_\ell$ Deletion}. In these problems, one is given a~$c$-edge-colored graph and wants to destroy all induced $c$-colored paths or cycles, respectively, on~$\ell$ vertices by deleting at most $k$ edges. Herein, a path or cycle is $c$-colored if it contains edges of $c$ distinct colors. We show that \textsc{$c$-Colored $P_\ell$ Deletion} and \textsc{$c$-Colored $C_\ell$ Deletion} are NP-hard for each non-trivial combination of $c$ and $\ell$. We then analyze the parameterized complexity of these problems. We extend the notion of neighborhood diversity to edge-colored graphs and show that both problems are fixed-parameter tractable with respect to the \emph{colored neighborhood diversity} of the input graph. We also provide hardness results to outline the limits of parameterization by the standard parameter solution size~$k$. Finally, we consider  bicolored input graphs and show a special case of  \textsc{$2$-Colored $P_4$ Deletion} that can be solved in polynomial time.
\end{abstract}

\section{Introduction}
A classic type of graph problems are \emph{edge-deletion} problems,
 where one wants to modify a given graph such that it fulfills some graph property $\Pi$ using a
minimum number of edge deletions.
%
While edge-deletion problems are well-studied on simple, uncolored graphs~\cite{Aravind17,Burzyn06,CDFG20,Mallah88,Yannakakis81}, there is not much work on such problems on edge-colored graphs.
At the same time, edge-colored graphs have an increasing number of applications, for example as a formal model of multilayer networks, an important tool for describing complex systems with multiple types of relations~\cite{Berlingerio11,Boccaletti14,SK20}.
Motivated by this, we analyze the computational complexity of the following fundamental edge-deletion problem on edge-colored graphs. 


\begin{quote}
\textsc{$c$-Colored $P_\ell$ Deletion} (\textsc{$cP_\ell$D})\\
\textbf{Input}:	A~$c$-edge-colored graph $G=(V,E=E_1\sqcup\ldots\sqcup E_c)$\footnotemark , an integer~$k$.\footnotetext{The $\sqcup$-operator denotes the disjoint union. Hence, $E_i \cap E_j =\emptyset$ for $i,j\in[c]$, $i\neq j$.}\\
\textbf{Question}: Is there a set~$S$ of at most $k$ edges such that deleting~$S$ from~$G$ results in a graph that contains no $c$-colored $P_\ell$ as induced subgraph?
\end{quote}

Herein, a path is \emph{$c$-colored} if its edge set is colored by exactly~$c$ distinct colors. The set~$S$ is called a \emph{solution}. We also consider \textsc{$c$-Colored $C_\ell$ Deletion} (\textsc{$cC_\ell$D}) where we aim to destroy all $c$-colored induced cycles of $\ell$ vertices  for a fixed $c\in[\ell]$. Moreover, several of our results hold also for the variant where we aim to destroy non-induced $c$-colored paths or cycles. Let us remark that we are not aware of a direct application of these problems. It seems, however, very plausible that algorithmic knowledge about these fundamental edge-deletion problems will be useful in more applied settings. They may, for example, serve as a starting point of showing hardness of other edge-modification problems in edge-colored graphs.

\paragraph{Related Work.}
The uncolored case of \textsc{$cP_\ell$D} and \textsc{$cC_\ell$D} is NP-hard for any~$\ell\geq 3$~\cite{Mallah88,Yannakakis81}.
In fact, it is known that \textsc{$H$ Deletion} is NP-complete for any graph~$H$ if and only if $H$ has at least two edges~\cite{Aravind17}. 
Recently, it was shown that \textsc{2$P_3$D} is NP-hard as well~\cite{gruettemeier19}.
If the input is restricted to graphs~$G$ where each bicolored $P_3$ of~$G$ is an induced subgraph of~$G$, then~\textsc{2$P_3$D} is polynomial-time solvable~\cite{gruettemeier19}.
Furthermore, the problem of destroying not only induced, but also non-induced bicolored paths containing three vertices is polynomial-time solvable on bicolored graphs but NP-hard on 3-colored graphs~\cite{Cai18}. The problem of detecting paths and cycles with certain edge-colorings has also received a considerable amount of attention~\cite{ADF+08,BWZ05,GLMMP09,GLMM12}. Considering vertex-colored graphs instead of edge-colored graphs, Bruckner et al.~\cite{BHKNTU12} studied the parameterized complexity of an edge-deletion problem, where one aims to obtain a graph in which no connected component contains two vertices of the same color.

\paragraph{Our Results.}
First, we consider the classical complexity of \textsc{$cP_\ell$D} and \textsc{$cC_\ell$D} in Section~\ref{sec:NP}.
We show that \textsc{$cP_\ell$D} is NP-hard for each $\ell\geq 3$ and each $c\in [\ell-1]$ and that \textsc{$cC_\ell$D} is NP-hard for each $\ell\geq 3$ and each $c\in [\ell]$.
Since \textsc{$cP_\ell$D} is trivially solvable if $\ell<3$ or $c\geq\ell$, and \textsc{$cC_\ell$D} is not properly defined if~$\ell<3$ and trivially solvable if $c>\ell$, this implies that \textsc{$cP_\ell$D} and \textsc{$cC_\ell$D} are NP-hard for each non-trivial combination of $c$ and $\ell$.
An aspect that complicates the analysis of edge-deletion problems is that an edge deletion may create a new forbidden induced subgraph. 
To formalize this effect, we introduce the following notion. 
An input graph~$G$ of an instance $(G,k)$ of \textsc{$cP_\ell$D} (\textsc{$cC_\ell$D}) is \textit{strictly non-cascading} if each subgraph of $G$ that is a~$c$-colored path (a $c$-colored cycle)  is an \emph{induced} subgraph.
Using these terms, \textit{2$P_3$D} is polynomial-time solvable on strictly non-cascading graphs~\cite{gruettemeier19}.
We show that \textit{$cP_\ell$D} is NP-hard on strictly non-cascading graphs for each $\ell\geq 4$ and each $c\in[2,\ell-2]$ and that \textsc{$cC_\ell$D} is NP-hard on strictly non-cascading graphs for each $\ell\geq 3$ and each $c\in[\ell]$.
\todo[inline]{Should we also discuss \emph{non-cascading} here?}

Second, we consider the parameterized complexity of \textsc{$cP_\ell$D} 
in Section~\ref{sec:fpt}. We consider a new parameter that we call the \emph{colored neighborhood diversity}. This parameter measures the number of sets of vertices that have different colored neighborhoods in the edge-colored graph. 
We show that \textsc{$cP_\ell$D} and \textsc{$cC_\ell$D} are fixed-parameter tractable with respect to this parameter. 
We believe that colored neighborhood diversity may be of broader interest in the study of computational problems on edge-colored input graphs. Furthermore, we study parameterization by the standard parameter solution size~$k$.
By using standard search tree techniques for uncolored graphs~\cite{Cai96}, it is easy to see that \textsc{$cP_\ell$D} and \textsc{$cC_\ell$D} are fixed-parameter tractable with respect to~$k$ for any fixed $\ell\geq 1$ and any~${c\in[\ell-1]}$ or $c\in[\ell]$, respectively. 
We prove that for \textsc{$cP_\ell$D} (or \textsc{$cC_\ell$D}, respectively) this cannot be improved to a running time where the exponential factor is independent of $\ell$ unless FPT=W[2] even if forbidden subgraphs can be detected in polynomial time.

Finally, we study $2P_4$D on subclasses of bicolored input graphs. 
We show that $2P_4$D remains NP-hard even when each edge color induces a cluster graph. 
If, additionally, the input graph contains no induced bicolored $P_4$ that starts and ends with edges of the same color, then $2P_4$D can be solved in polynomial time. 
The algorithm is based on a characterization of such graphs that may be of independent interest.    

\section{Preliminaries}

\label{sec:preleminaries}
For integers~${a\in\mathbb{N}}$ and $b\in\mathbb{N}$ we define $[a,b]:=\{a,\dots,b\}$ if~$a \leq b$ and~$[a,b]= \emptyset$ if~$a>b$. Furthermore, we define~$[b]:=[1,b]$. 
We denote the vertex set of a graph~$G$ by $V(G)$ and its edge set by $E(G)$. Throughout this work let~$n:=|V(G)|$ and~$m:=|E(G)|$.  
By~${|G|:=n+m}$ we denote the \textit{size} of $G$.
In contrast to a graph, the edges of a \textit{multigraph} are a multiset. Hence, in a multigraph two vertices can have multiple edges between them.
In an \emph{edge-colored graph}, the edge set is partitioned into $c$ disjoint, non-empty subsets $E_1,\dots,E_c$. 
We call such a graph \textit{$c$-colored} and if~${c=1}$ we call the graph \textit{uncolored}.
For sake of illustration we define $E_b:=E_1$ as the set of \textit{blue} edges, $E_r:=E_2$ as the set of \textit{red} edges,~$E_y:=E_3$ as the set of \textit{yellow} edges, $E_g:=E_4$ as the set of \textit{green} edges, and draw them in the figures accordingly.
For a given graph $G$ we denote the set of all edges with color $\alpha\in [c]$ by~$E_\alpha(G)$.

For two vertex sets~$V'\subseteq V(G)$ and $V''\subseteq V(G)$ we define~$E_G(V',V''):=\{\{u,v\}\in E\ |\ u\in V'\text{ and } v\in V''\}$ and~$E_G(V'):=E_G(V',V')$. 
We call $G':=(V',E')$ a \textit{subgraph} of $G$ if $V'\subseteq V(G)$ and~$E'\subseteq E_G(V')$.
If~${E'= E_G(V')}$ we call~$(V',E')$ an \textit{induced} subgraph and write~$G[V']$.
 We denote the graph formed by deleting the edges of an edge set $E'$ from $G$ by~${G-E' := (V,E\setminus E')}$.
For a vertex $v\in V(G)$, we denote the \textit{open neighborhood} of $v$ in $G$ by~$N_{G}(v) := \{u\in V(G) \mid \{u,v\}\in E(G)\}$ and the~\emph{closed neighborhood} by~$N_G[v] := N_G(v) \cup \{v\}$.
We denote the \textit{degree} of $v$ in $G$ by $\deg_{G}(v) := |N_G(v)|$. For every color~$i \in [c]$ we define the~$i$-neighborhood by~$N^i_G(v):=\{u \in V(G) \mid \{u,v\} \in E_i\}$. 
We may drop the subscript~$\cdot_G$ when it is clear from context.

We say that a vertex set $V'\subseteq V(G)$ is a \textit{vertex cover} for $G$ if at least one of $u$ and $v$ is in $V'$ for each edge~${\{u,v\}\in E(G)}$.
We say that a vertex set $V'\subseteq V(G)$ is an \textit{independent set} if $\{u,v\}\notin E(G)$ for each pair of vertices $u,v\in V'$.
A graph $G$ is called \textit{tripartite} if $V(G)$ can be partitioned into 3 (possibly empty) independent sets.

A graph~$G$ is a \textit{path} if it is possible to index its vertices with numbers from~$[n]$, in such a way that~$\{v_i,v_j\}\in E(G)$ if and only if $i+1=j$. 
A path with the additional edge $\{v_n,v_1\}$ is called a \textit{cycle}.
We denote a path (cycle) consisting of $\ell$ vertices by $P_\ell$ ($C_\ell$).
The length of a path (cycle) $G$ is the number of its edges,~$|E(G)|$. 
For any graph $G$, the 2-\textit{subdivision} of $G$ is the graph we get from inserting two new vertices on every edge, that is, from replacing each edge~${\{u,v\}\in E(G)}$ by a $P_4$ with vertices~$\{u,x,y,v\}$.
A graph $G$ is a 2-\textit{subdivision graph} if it is the 2-subdivision of some graph~$H$.  
 The \textit{girth} of~$G$ is the length of a shortest cycle in $G$, for acyclic graphs the girth is infinite. 
 
For two edge-colored graphs $G$ and $F$, we say that $G$ is \textit{isomorphic} to $F$, and we write~${G\cong F}$ if there is a bijective function $\varphi:V(G)\rightarrow V(F)$ such that~${\{u,v\}\in E(G)}$ is an edge with color $c$ if and only if~$\{\varphi(u),\varphi(v)\}\in E(F)$ is an edge with color~$c$. Furthermore, we say that $G$ is $F$-\textit{free} if $G[V']\ncong F$ for each vertex set $V'\subseteq V(G)$. 

In this work we study the computational problems~$cP_\ell$D and~$cC_\ell$D. We sometimes use the following problem formulation that generalizes $cP_\ell$D and~$cC_\ell$D. 
Herein, we let~$\mathcal{F}$ is a set of~$c$-colored graphs. 
\begin{quote}
\textsc{$\mathcal{F}$-Deletion}\\
\textbf{Input}:	An edge-colored graph $G=(V,E=E_1\sqcup\ldots\sqcup E_c)$, an integer~$k$.\\
\textbf{Question}: Is there a set~$S$ of at most $k$ edges such that~$G-S$ is~$F$-free for every~$F \in \mathcal{F}$?
\end{quote}

Observe that, for given~$c$ and~$\ell$ we can define~$\mathcal{F}$ as the set of~$c$-colored~$P_\ell$ (or~$c$-colored~$C_\ell$, respectively). Thus, \textsc{$\mathcal{F}$-Deletion} generalizes the problems we consider in this work.

For the relevant notions of
parameterized complexity refer to the standard textbook~\cite{Cyg+15}.

\section{Classical Complexity} 
\label{sec:NP}
We first prove the NP-hardness of \textsc{$cP_\ell$D} and \textsc{$cC_\ell$D}. In our analysis we focus on the impact of cascading effects on the complexity: The NP-hardness of \textsc{$2P_3$D} relies on the fact that edge deletions may create new bicolored~$P_3$s, as \textsc{$2P_3$D} is polynomial-time solvable on graphs that do not provide this cascading effect~\cite{gruettemeier19}. Here, we show that \textsc{$cP_\ell$D} with~$\ell \geq 4$ and~$c \in [2,\ell-2]$ is NP-hard even when limited to instances where~$G$ is non-cascading, defined as follows.

\begin{definition}
\begin{itemize}

\item
Let~$\mathcal{F}$ be a set of forbidden subgraphs.  An edge~$e$ of a graph~$G$ is a \emph{conflict edge} if there is some~$F \in \mathcal{F}$ such that~$e$ is contained in an induced~$F$ in~$G$.

\item
Let~$X$ be the set of all conflict edges. The graph~$G$ is \emph{non-cascading} if every non-induced subgraph~$F \in \mathcal{F}$ in~$G$ is not an induced~$F$ in~$G-X'$ for every subset~$X' \subseteq X$.

\item
The graph~$G$ is \emph{strictly non-cascading} if there is no non-induced subgraph~$F \in \mathcal{F}$ in~$G$.
\end{itemize}
\end{definition}

It is easy to see that a graph is non-cascading if it is strictly non-cascading. We call an edge \emph{conflict-free} if it is not a conflict edge.
Observe that in case of~$cP_\ell$D it can be checked in polynomial time if a given graph~$G$ is non-cascading since~$\ell$ is a constant: 
First, determine all conflict edges of~$G$ in $\Oh(n^\ell)$~time. 
Second, iterate over all~$\Oh(n^\ell)$ non-induced~$c$-colored~$P_\ell$ subgraphs of~$G$. Let~$v_1, \dots, v_\ell$ be the vertices of one such subgraph. Then, check if there is a conflict-free edge~$\{v_i,v_j\}$ with~$i,j \in [\ell]$ and~$j \neq i+1$. If this is the case for all non-induced~$c$-colored~$P_\ell$ subgraphs of~$G$ return \textit{yes}. Otherwise, return \textit{no}.
The main idea of non-cascading graphs is that it is sufficient to hit all initial conflicts in the input graph and that conflict-free edges are never part of a solution. 
The next proposition formalizes this idea.

\begin{proposition} \label{Prop: Non-Cascading}
Let~$(G,k)$ be an instance of \textsc{$\mathcal{F}$-Deletion} where~$G$ is non-cascading, and let~$X$ be the set of conflict edges of~$G$. If~$\tilde{S}$ is an edge-deletion set such that every induced~$F \in \mathcal{F}$ in~$G$ is not an induced~$F$ in~$G-\tilde{S}$, then~$S:=\tilde{S} \cap X$ is an edge-deletion set such that~$G-S$ has no induced~$F \in \mathcal{F}$.
\end{proposition}

\begin{proof}
Assume towards a contradiction that~$G-S$ contains an induced~$F \in \mathcal{F}$. Let~$e_1, \dots, e_\ell$ be the edges of an induced~$F$ in~$G-S$. Since~$G-\tilde{S}$ contains no induced~$F$ one of the edges~$e_i$ with~$i \in [\ell]$ belongs to~$\tilde{S} \setminus X$ and therefore~$e_i$ is conflict-free. 
Consequently,~$e_1, \dots, e_\ell$ form a non-induced~$F$ in~$G$. 
Then, since~$G$ is non-cascading and~$S \subseteq X$ it follows that~$e_1, \dots, e_\ell$ do not form an induced~$F$ in~$G-S$ which contradicts our assumption.
\end{proof}

Observe that Proposition~\ref{Prop: Non-Cascading} implies that a minimal solution only consists of conflict edges if the input graph is non-cascading.


\subsection{$c$-colored $P_\ell$ Deletion}
\label{sec:pld}
First, we show the hardness of~\textsc{$cP_\ell$D}.
The following lemma is useful for showing the correctness of our reduction.

\begin{lemma}
Let $G=(V=\{v_1,\dots,v_{d\cdot (\ell-1)}\},E)$ be a path where any $\ell$ consecutive vertices form a $c$-colored $P_\ell$ and let~$S\subseteq E$ be an edge set of size $d-1$. Then, $G- S$~is $c$-colored $P_\ell$-free if and only if~$S=\{\{v_i,v_{i+1}\}\ |\ i\mod~ (\ell-1) = 0\}$.
\label{lem:path}
\end{lemma}
\begin{proof}
Let $S=\{\{v_i,v_{i+1}\}\ |\ i\mod (\ell-1) = 0\}$.
Since $|V|=d\cdot (\ell-1)$ we conclude that $|S|=d-1$.
We will show that $G':=G- S$ is $P_\ell$-free and hence, $G'$ is $c$-colored $P_\ell$-free.
Let~${V'\subseteq V}$ such that the induced subgraph $G[V']$ is a $c$-colored $P_\ell$.
Since $G$ is a path, $V'=\{v_i,\dots,v_{i+\ell-1}\}$ for some~$i\in[(d-1)\cdot(\ell-1)]$.
Since $|[i,i+\ell-2]|=\ell-1$ there is a $\hat{i}\in [i,i+\ell-2]$ such that~$\hat{i}\mod (\ell-1) = 0$.
Thus, we conclude that there is an edge $e\in E_G(V')$ such that $e\in S$.
Hence, $G'$ is $c$-colored $P_\ell$-free.

Conversely, let $S$ be an edge-deletion set of size $d-1$ such that $G-S$ is $c$-colored $P_\ell$-free.
Let~$\{v_i,v_{i+1}\}$ be the $j$-th edge in $S$. 
We denote~$V_1:=\{v_1,\dots,v_i\}$ and $V_2:=\{v_{i+1},\dots,v_{d\cdot(\ell-1)}\}$.
Since~$\{v_i,v_{i+1}\}$ is the $j$-th edge in $S$, we know that~${|S\cap E_G(V_1)|=j-1}$.

First, assume towards a contradiction that $i\geq j\cdot(\ell-1)+1$.
Then, 
\begin{equation*}
|V_1| = i\geq j\cdot(\ell-1)+1.
\end{equation*}
Since any $\ell$ consecutive vertices form a $c$-colored $P_\ell$, we conclude that $G[V_1]$ contains at least $j$ edge-disjoint $c$-colored $P_\ell$s.
Since $|S\cap E_G(V_1)|=j-1$, this is a contradiction to the condition that $G-S$ is $c$-colored $P_\ell$-free.
Thus, $i\leq j\cdot(\ell-1)$.

Next, assume towards a contradiction that $i\leq j\cdot(\ell-1)-1$.
Then, \begin{align*}
|V_2|&=d\cdot(\ell-1)-i\\&\geq (d-j)\cdot(\ell-1)+1.
\end{align*}
Since any $\ell$ consecutive vertices form a $c$-colored $P_\ell$, we conclude that $G[V_2]$ contains at least $(d-j)$ edge-disjoint $c$-colored $P_\ell$s.
Hence, $|S\cap E_G(V_2)|\geq (d-j)$. 
We conclude that $|S|\geq d$, since~$|S\cap E_G(V_1)|=j-1$ and $\{v_i,v_{i+1}\}\in S$.
This is a contradiction to the condition that $|S|=d-1$.
Thus,~$i\geq j\cdot(\ell-1)$.

Hence, we know that $i=j\cdot(\ell-1)$. Thus, the $j$-th edge in $S$ is $\{v_{j\cdot(\ell-1)},v_{j\cdot(\ell-1)+1}\}$ and therefor~$S=\{\{v_i,v_{i+1}\}\ |\ i\mod (\ell-1) = 0\}$.
\end{proof}

Now we will show that \textsc{$cP_\ell$D} is NP-hard for $\ell\geq 4$ even if the input graph has a somewhat simple structure.
From this result we will be able to prove that \textsc{$cP_\ell$D} remains NP-hard on non-cascading input graphs if $\ell\geq 4$ and~${c\in[2,\ell-2]}$.

\begin{theorem}
\textsc{$cP_\ell$D} is NP-hard for each~$\ell\geq 4$ and each~${c\in[2,\ell-2]}$ even if the maximum degree of~$G$ is three, and the girth of $G$ is greater than $2\cdot d\cdot\ell$ for any constant $d\geq 1$.
\label{thm:cPD}
\end{theorem}
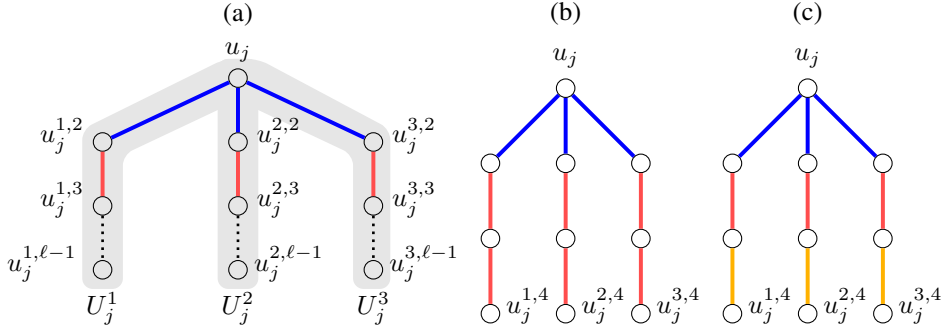
\begin{figure}[t!]
\centering
\subfloat{
\centering
\begin{tikzpicture}[xscale=0.9,yscale=0.85]

\node at (2,4){(a)};
\begin{pgfonlayer}{background}

\begin{scope}[blend mode=multiply]
      \draw[rounded corners, fill=black!10,draw=none] (2.3,3.3)--(1.8,3.3)--(-.3,2.2)--(-.3,-.3)--(.3,-.3)--(.3,1.8)--(2.3,2.8)--cycle;
      \draw[rounded corners, fill=black!10,draw=none] (3.7,-.3)--(4.25,-.3)--(4.25,2.2)--(2.3,3.3)--(1.7,3.3)--(1.7,2.8)--(3.7,1.8)--cycle;

      \draw[rounded corners, fill=black!10,draw=none](2.3,-.3)--(1.7,-.3)--(1.7,3.3)--(2.3,3.3)--cycle;
  
\end{scope}
    \end{pgfonlayer}
    \node at (0,-.6){$U_j^1$};
\node at (2,-.6){$U_j^2$};
\node at (4,-.6){$U_j^3$};

\node[vertex](u) at (2,3){};\node at (2,3.4){$u_j$};
\node[vertex](11) at (0,2){};\node at (-.6,2.1){$u_j^{1,2}$};
\node[vertex](21) at (2,2){};\node at (2.6,2.1){$u_j^{2,2}$};
\node[vertex](31) at (4,2){};\node at (4.6,2.1){$u_j^{3,2}$};
\node[vertex](12) at (0,1){};\node at (-.6,1.1){$u_j^{1,3}$};
\node[vertex](22) at (2,1){};\node at (2.6,1.1){$u_j^{2,3}$};
\node[vertex](32) at (4,1){};\node at (4.6,1.1){$u_j^{3,3}$};
\node[vertex](13) at (0,0){};\node at (-.9,.1){$u_j^{1,\ell-1}$};
\node[vertex](23) at (2,0){};\node at (2.75,.1){$u_j^{2,\ell-1}$};
\node[vertex](33) at (4,0){};\node at (4.75,.1){$u_j^{3,\ell-1}$};

\path[blue edge](u)--(11);
\path[blue edge](u)--(21);
\path[blue edge](u)--(31);
\path[red edge](12)--(11);
\path[red edge](22)--(21);
\path[red edge](32)--(31);
\path[black dots](12)--(13);
\path[black dots](22)--(23);
\path[black dots](32)--(33);
\end{tikzpicture}
}
\subfloat{
\centering
\begin{tikzpicture}[scale=1]
  \node at (1,4){(b)};

\node[vertex](u) at (1,3){};\node at (1,3.4){$u_j$};
\node[vertex](11) at (0,2){};
\node[vertex](21) at (1,2){};
\node[vertex](31) at (2,2){};
\node[vertex](12) at (0,1){};
\node[vertex](22) at (1,1){};
\node[vertex](32) at (2,1){};
\node[vertex](13) at (0,0){};\node at (.5,.1){$u_j^{1,4}$};
\node[vertex](23) at (1,0){};\node at (1.5,.1){$u_j^{2,4}$};
\node[vertex](33) at (2,0){};\node at (2.5,.1){$u_j^{3,4}$};

\path[blue edge](u)--(11);
\path[blue edge](u)--(21);
\path[blue edge](u)--(31);
\path[red edge](12)--(11);
\path[red edge](22)--(21);
\path[red edge](32)--(31);
\path[red edge](12)--(13);
\path[red edge](22)--(23);
\path[red edge](32)--(33);
\end{tikzpicture}
}
\subfloat{
\centering
\begin{tikzpicture}[scale=1]
\node at (1,4){(c)};
\node[vertex](u) at (1,3){};\node at (1,3.4){$u_j$};
\node[vertex](11) at (0,2){};
\node[vertex](21) at (1,2){};
\node[vertex](31) at (2,2){};
\node[vertex](12) at (0,1){};
\node[vertex](22) at (1,1){};
\node[vertex](32) at (2,1){};
\node[vertex](13) at (0,0){};\node at (.5,.1){$u_j^{1,4}$};
\node[vertex](23) at (1,0){};\node at (1.5,.1){$u_j^{2,4}$};
\node[vertex](33) at (2,0){};\node at (2.5,.1){$u_j^{3,4}$};

\path[blue edge](u)--(11);
\path[blue edge](u)--(21);
\path[blue edge](u)--(31);
\path[red edge](12)--(11);
\path[red edge](22)--(21);
\path[red edge](32)--(31);
\path[yellow edge](12)--(13);
\path[yellow edge](22)--(23);
\path[yellow edge](32)--(33);
\end{tikzpicture}
}
\caption{(a) General structure of a clause gadget $Z_j$. The black dotted line represents a path containing $\ell -3$ vertices in total. (b) Clause gadget for $\ell = 5$, $c=2$. (c) Clause gadget for $\ell=5$, $c=3$. }
\label{fig:cla_gadget}
\end{figure}

\begin{proof}
We give a polynomial-time reduction from the NP-complete \textsc{(3,B2)-SAT} problem~\cite{berman03}, a version of \textsc{3SAT} where one is given a CNF formula~$\Phi$ on variables~$x_1, \dots, x_\eta$ where every clause contains exactly three literals and each literal~$x_i$ and~$\neg x_i$ occurs exactly twice in~$\Phi$.

\textit{Construction}:
Let~$\Phi$ be a \textsc{(3,B2)-SAT} formula with clauses~${\mathcal{C}=\{c_1,\dots,c_{\mu}\}}$ and variables $\mathcal{X}=\{x_1,\dots,x_{\eta}\}$.  
We use the following \textit{gadgets} to construct an equivalent instance $(G=(V,E),k)$ of \textsc{$cP_\ell$D} from~$\Phi$.

For each clause~$c_j\in\mathcal{C}$ we construct a~\textit{clause gadget}~$Z_j$  consisting of three vertex sets~${U_j^1, U_j^2}$ and~$U_j^3$, each containing $\ell-1$ vertices. 
For each~${p\in[3]}$, we denote the vertices in~$U_j^p$ by~${u_j^{p,1},\dots,u_j^{p,\ell-1}}$.
For~$s\in[\ell-2]$ we add edges~${\{u_j^{p,s}, u_j^{p,s+1}\}}$.
If $s<c$ we add an edge of color $s$.
Else we add an edge of color $c$.
In other words~${\{u_j^{p,1}, u_j^{p,2}\}}$ is blue, $\{u_j^{p,2}, u_j^{p,3}\}$ is red and the color of the next edges depends on $c$.
Observe that $G[U_j^p]$ is a $c$-colored~$P_{\ell-1}$.
We connect the three~$P_{\ell-1}$s by identifying~$u_j^{1,1}=u_j^{2,1}=u_j^{3,1}=:u_j$ (see Fig.~\ref{fig:cla_gadget} for an example).

For each variable~${x_i\in\mathcal{X}}$ we construct a \textit{variable gadget}~$X_i$ as follows. 
First, let~${z:=d\cdot(\ell-1)+1}$.
Note that $z$ is the minimum number of vertices on a path that contains $d$ edge-disjoint~$P_\ell$s.
The variable gadget~$X_i$ consists of four vertex sets of $z$ vertices~${T_i^1,T_i^2,F_i^1,F_i^2}$ and a vertex set of $\ell-4$ vertices $W_i:=\{w_i^1,\dots,w_i^{\ell-4}\}$.
Note that $W_i=\emptyset$ for $\ell = 4$.
For~${q\in[2]}$ the vertices in $T_i^q$ are denoted by $t_i^{q,1},\dots,t_i^{q,z}$ and the vertices in $F_i^q$ are denoted by $f_i^{q,1},\dots,f_i^{q,z}$.
For $s\in[z-1]$ we add edges~${\{t_i^{q,s}, t_i^{q,s+1}\}}$ and~${\{f_i^{q,s}, f_i^{q,s+1}\}}$.
If $0< s\mod (\ell-1)< c$ we add an edge with color $s\mod (\ell-1)$ and else we add an edge with color $c$.
So if $s\mod (\ell-1)=1$ we add a blue edge, if $s\mod (\ell-1)=2$ we add a red edge, and otherwise the edge-color depends on $c$.

We connect $T_i^1$ and $T_i^2$ by identifying $t_i^{1,1}=t_i^{2,1}=:t_i$ and analogously we connect~$F_i^1$ and~$F_i^2$ by identifying $f_i^{1,1}=f_i^{2,1}=:f_i$.
If $\ell = 4$ we add a red edge~$\{t_i,f_i\}$.
If~$\ell > 4$ we connect $t_i$ and $f_i$ by a $(c-1)$-colored path with vertices in $\{t_i\}\cup W_i\cup\{f_i\}$ that does not contain a blue edge (see Fig.~\ref{fig:var_gadget}).
Since $|\{t_i\}\cup W_i\cup\{f_i\}|=\ell-2\geq c$, this path always has at least $(c-1)$ edges.
Hence, we can always connect $t_i$ and $f_i$ by a $(c-1)$-colored path.

Note that any $\ell$ consecutive vertices in $G[T_i^1]$, $G[T_i^2]$, $G[F_i^1]$ and $G[F_i^2]$ form a $c$-colored $P_\ell$.
Hence, $G[T_i^1]$, $G[T_i^2]$, $G[F_i^1]$ and $G[F_i^2]$ are four paths, each containing~$d$ edge-disjoint $c$-colored $P_\ell$s. 
And since~$\{t_i^{q,2},t_i\}\in E_b(G)$ and $t_i$ and~$f_i$ are connected by $(c-1)$-colored $P_{\ell-1}$ with no blue edge, the induced subgraphs $G[\{t_i^{q,2},t_i,w_i^1,\dots,w_i^{\ell -4},f_i,f_i^{q',2}\}]$ are also $c$-colored $P_\ell$s for $q,q'\in[2]$. 

Then, we denote the following edge sets (see Fig.~\ref{fig:var_gadget}).
\begin{itemize}
\item[$T_i^b$]$:=\{\{t_i^{q,s},t_i^{q,s+1}\}\in E\ |\ s\mod (\ell-1)=1,q\in[2]\}$.
\item[$F_i^b$]$:=\{\{f_i^{q,s},f_i^{q,s+1}\}\in E\ |\ s\mod (\ell-1)=1,q\in[2]\}$.
\item[$T_i^r$]$:=\{\{t_i^{q,s},t_i^{q,s+1}\}\in E\ |\ s\mod (\ell-1)=2,q\in[2]\}$.
\item[$F_i^r$]$:=\{\{f_i^{q,s},f_i^{q,s+1}\}\in E\ |\ s\mod (\ell-1)=2,q\in[2]\}$.
\end{itemize}
Note that $|T_i^r|=|F_i^r|=|T_i^b|=|F_i^b|=2\cdot d$.

\begin{figure}[t!]
\begin{center}
\subfloat{
\centering
\begin{tikzpicture}[xscale=.83, yscale=0.62]

  \node at (6.5,3.7){(a)};
  \begin{pgfonlayer}{background}
    \begin{scope}[blend mode=multiply] 
\draw[rounded corners, fill=black!15,draw=none](-.4,-.3)--(3.3,-.3)--(4.3,.7)--(4.3,1.3)--(3.8,1.3)--(2.8,.4)--(-.4,.4)--cycle;
\draw[rounded corners, fill=black!15,draw=none](-.4,1.7)--(2.8,1.7)--(3.8,.7)--(4.3,.7)--(4.3,1.3)--(3.3,2.4)--(-.4,2.4)--cycle;
\draw[rounded corners, fill=black!15,draw=none](5,.7)--(8,.7)--(8,1.3)--(5,1.3)--cycle;
\draw[rounded corners, fill=black!15,draw=none](13.4,-.3)--(9.7,-.3)--(8.7,.7)--(8.7,1.3)--(9.2,1.3)--(10.2,.4)--(13.4,.4)--cycle;
\draw[rounded corners, fill=black!15,draw=none](13.4,1.7)--(10.2,1.7)--(9.2,.7)--(8.7,.7)--(8.7,1.3)--(9.7,2.4)--(13.4,2.4)--cycle;
     
    \end{scope}
  \end{pgfonlayer}
\node at(1.5,-.7){$T_i^2$};
\node at(1.5,3.3){$T_i^1$};
\node at(11.5,-.7){$F_i^2$};
\node at(6.5,2.1){$W_i$};
\node at(11.5,3.3){$F_i^1$};

\node[vertex](t1z)at(0,2){};\node at(0,2.75){$t_i^{1,z}$};
\node[vertex](t12)at(3,2){};\node at(3,2.75){$t_i^{1,2}$};
\node[vertex](t2z)at(0,0){};\node at(0,-.65){$t_i^{2,z}$};
\node[vertex](t22)at(3,0){};\node at(3,-.65){$t_i^{2,2}$};
\node[vertex](t)at(4,1){};\node at(4.3,1.7){$t_i$};
\node[vertex](w1)at(5.5,1){};\node at(5.5,0.3){$w_i^1$};
\node[vertex](w2)at(7.5,1){};\node at(7.5,0.3){$w_i^{\ell-4}$};
\node[vertex](f)at(9,1){};\node at(8.7,1.7){$f_i$};
\node[vertex](f1z)at(13,2){};\node at(13,2.75){$f_i^{1,z}$};
\node[vertex](f12)at(10,2){};\node at(10,2.75){$f_i^{1,2}$};
\node[vertex](f2z)at(13,0){};\node at(13,-.65){$f_i^{2,z}$};
\node[vertex](f22)at(10,0){};\node at(10,-.65){$f_i^{2,2}$};

\path[black dots](t1z)--(t12);
\path[black dots](t2z)--(t22);
\path[black dots](f1z)--(f12);
\path[black dots](f2z)--(f22);
\path[black dots](w1)--(w2);
\path[blue edge](t)--(t12);
\path[blue edge](t)--(t22);
\path[blue edge](f)--(f12);
\path[blue edge](f)--(f22);
\path[red edge](t)--(w1);
\path[black edge](w2)--(f);
\end{tikzpicture}
}
\vspace{-0.5em}

\subfloat{
\centering
\begin{tikzpicture}[scale=.62]

  \node at (9,3.7){(b)};

  \def\x{2}
\def\y{1}
\def\z{0}

\node[svertex](t) at (8,\y) {};
\node[svertex](w) at (9,\y) {};
\node[svertex](f) at (10,\y) {};
\path[red edge] (t)--(w);
\path[yellow edge] (w)--(f);

\node[svertex](t9) at (0,\x) {};
\node[svertex](t8) at (1,\x) {};
\node[svertex](t7) at (2,\x) {};
\node[svertex](t6) at (3,\x) {};
\node[svertex](t5) at (4,\x) {};
\node[svertex](t4) at (5,\x) {};
\node[svertex](t3) at (6,\x) {};
\node[svertex](t2) at (7,\x) {};

\path[yellow edge] (t9)--(t8);
\path[yellow edge] (t8)--(t7);
\path[red edge] (t7)--(t6);
\path[blue edge] (t6)--(t5);
\path[yellow edge] (t5)--(t4);
\path[yellow edge] (t4)--(t3);
\path[red edge] (t3)--(t2);
\path[blue edge] (t2)--(t);

\node[svertex](t9) at (0,\z) {};
\node[svertex](t8) at (1,\z) {};
\node[svertex](t7) at (2,\z) {};
\node[svertex](t6) at (3,\z) {};
\node[svertex](t5) at (4,\z) {};
\node[svertex](t4) at (5,\z) {};
\node[svertex](t3) at (6,\z) {};
\node[svertex](t2) at (7,\z) {};

\path[yellow edge] (t9)--(t8);
\path[yellow edge] (t8)--(t7);
\path[red edge] (t7)--(t6);
\path[blue edge] (t6)--(t5);

\node[label] at (0,\x){$t_i^{1,9}$};
\node[label] at (4,\x){$t_i^{1,5}$};
\node[label] at (0,\z){$t_i^{2,9}$};
\node[label] at (4,\z){$t_i^{2,5}$};
\node[label] at (8,\y){$t_i$};
\node[label] at (9,\y){$w_i^1$};
\node[label] at (10,\y){$f_i$};
\node[label] at (14,\x){$f_i^{1,5}$};
\node[label] at (18,\x){$f_i^{1,9}$};
\node[label] at (14,\z){$f_i^{2,5}$};
\node[label] at (18,\z){$f_i^{2,9}$};

\path[yellow edge] (t5)--(t4);
\path[yellow edge] (t4)--(t3);
\path[red edge] (t3)--(t2);
\path[blue edge] (t2)--(t);

\node[svertex](t2) at (11,\x) {};
\node[svertex](t3) at (12,\x) {};
\node[svertex](t4) at (13,\x) {};
\node[svertex](t5) at (14,\x) {};
\node[svertex](t6) at (15,\x) {};
\node[svertex](t7) at (16,\x) {};
\node[svertex](t8) at (17,\x) {};
\node[svertex](t9) at (18,\x) {};

\path[yellow edge] (t9)--(t8);
\path[yellow edge] (t8)--(t7);
\path[red edge] (t7)--(t6);
\path[blue edge] (t6)--(t5);
\path[yellow edge] (t5)--(t4);
\path[yellow edge] (t4)--(t3);
\path[red edge] (t3)--(t2);
\path[blue edge] (t2)--(f);

\node[svertex](t2) at (11,\z) {};
\node[svertex](t3) at (12,\z) {};
\node[svertex](t4) at (13,\z) {};
\node[svertex](t5) at (14,\z) {};
\node[svertex](t6) at (15,\z) {};
\node[svertex](t7) at (16,\z) {};
\node[svertex](t8) at (17,\z) {};
\node[svertex](t9) at (18,\z) {};
\path[yellow edge] (t9)--(t8);
\path[yellow edge] (t8)--(t7);
\path[red edge] (t7)--(t6);
\path[blue edge] (t6)--(t5);
\path[yellow edge] (t5)--(t4);
\path[yellow edge] (t4)--(t3);
\path[red edge] (t3)--(t2);
\path[blue edge] (t2)--(f);
\end{tikzpicture}
}

\end{center}
\caption{(a)~The generalized structure of a variable gadget $X_i$. Note that $W_i=\emptyset$ if $\ell=4$. The edges of color $c$ are black.  (b)
  ~An exemplary variable gadget for $\ell=5,c=3,d=2$.}
\label{fig:var_gadget}
\end{figure}
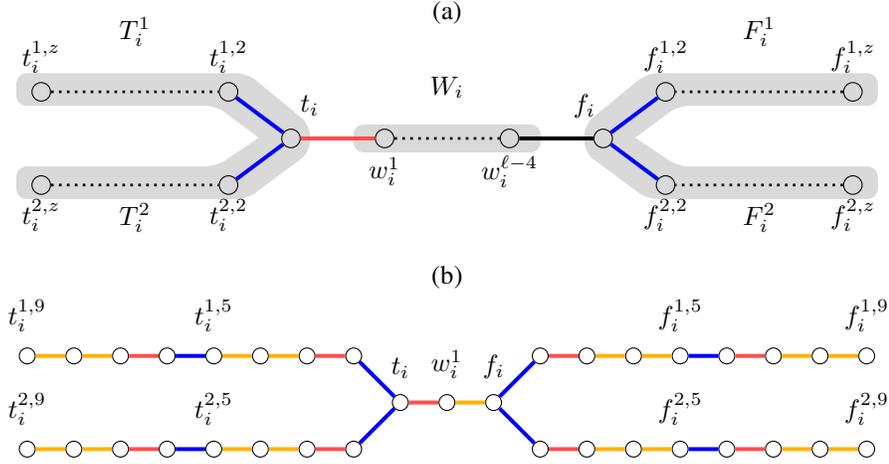

To connect the variable and clause gadgets we identify vertices as follows (see Fig.~\ref{fig:example}). 
For ${p\in[3]}$, ${q\in[2]}$, any variable~${x_i\in\mathcal{X}}$ and any clause~${c_j\in\mathcal{C}}$ we set 
\begin{center}
$u_j^{p,2}= 
\begin{cases}
t_i^{q,z}&\text{if the literal }x_i\text{ has its }q\text{-th occurence as the }p\text{-th literal in }c_j\\
f_i^{q,z}&\text{if the literal }\lnot x_i\text{ has its }q\text{-th occurence as the }p\text{-th literal in }c_j.
\end{cases}$
\end{center}

Note that for each clause gadget~$Z_j$ and $p\in[3]$ the vertex~$u_j^{p,2}$ is identified with exactly one vertex from a variable gadget and for each variable gadget~$X_i$ and $q\in[2]$ the vertices $t_i^{q,z}$ and $f_i^{q,z}$ are each identified with exactly one vertex from a clause gadget, since each literal occurs exactly twice in~$\Phi$. 

It is easy to see that the maximum degree of $G$ is three.
Furthermore, the girth of $G$ is at least $2\cdot d\cdot\ell$, since the smallest possible cycle in $G$ contains the vertices from $T_i^1\cup T_i^2\cup\{u_j\}$ or~$F_i^1\cup F_i^2\cup\{u_j\}$, respectively.
Such a cycle will be constructed, when there is a clause $c_j=(x_i\lor x_i\lor \dots)$ or $c_j=(\lnot x_i\lor \lnot x_i\lor \dots)$ in~$\Phi$.
We complete the construction by setting~${k:=4\cdot d\cdot \eta+2\cdot\mu}$.

\begin{figure}[t!]
\centering
\begin{tikzpicture}[xscale=0.72, yscale=0.58]

  \begin{pgfonlayer}{background}
    
  \begin{scope}[blend mode=multiply]
\draw[rounded corners, fill=black!10,draw=none](-2,-.5)--(2.5,-.5)--(2.5,8)--(-2,8)--cycle;\node at(0.2,8.3){$Z_j$};
\draw[rounded corners, fill=black!10,draw=none](3.4,-.5)--(12,-.5)--(12,1.5)--(.7,1.5)--(.7,.5)--(3.4,.5)--cycle;\node at(7.5,1.8){$X_3$};
\draw[rounded corners, fill=black!10,draw=none](3.4,2.5)--(12,2.5)--(12,4.5)--(.7,4.5)--(.7,3.5)--(3.4,3.5)--cycle;\node at(7.5,4.8){$X_2$};
\draw[rounded corners, fill=black!10,draw=none](3.4,8)--(12,8)--(12,5.5)--(.7,5.5)--(.7,6.5)--(3.4,6.5)--cycle;\node at(7.5,8.3){$X_1$};
    
  \end{scope}
\end{pgfonlayer}
  
\node[svertex](uj) at (-1,4){};\node at (-1.5,4){$u_j$};
\node[svertex](u11) at (1,6){};\node at (-.2,6.1){$u_j^{1,2}=t_1^{2,4}~~$};
\node[svertex](u12) at (2,7){};
\node[svertex](u21) at (1,4){};
\node[svertex](u22) at (2,5){};
\node[svertex](u31) at (1,1){};\node at (-.2,.9){$u_j^{3,2}=t_3^{1,4}~~$};
\node[svertex](u32) at (2,2){};
\path[blue edge](uj)--(u11);
\path[blue edge](uj)--(u21);
\path[blue edge](uj)--(u31);
\path[red edge](u11)--(u12);
\path[red edge](u21)--(u22);
\path[red edge](u31)--(u32);

\node[svertex](t13) at (5,6){};
\node[svertex](t12) at (6,6){};
\node[svertex](t24) at (4,7){};\node at (4,7.5){$t_1^{1,4}$};
\node[svertex](t23) at (5,7){};
\node[svertex](t22) at (6,7){};
\node[svertex](ti) at (7,6.5){};\node at (7,6.9){$t_1$};
\node[svertex](fi) at (8,6.5){};\node at (8,6.9){$f_1$};
\node[svertex](f12) at (9,6){};
\node[svertex](f13) at (10,6){};
\node[svertex](f14) at (11,6){};\node at (11.5,5.9){$f_1^{2,4}$};
\node[svertex](f22) at (9,7){};
\node[svertex](f23) at (10,7){};
\node[svertex](f24) at (11,7){};\node at (11.5,7.1){$f_1^{1,4}$};
\path[blue edge](ti)--(t12);
\path[blue edge](ti)--(t22);
\path[blue edge](fi)--(f12);
\path[blue edge](fi)--(f22);
\path[red edge](u11)--(t13);
\path[red edge](t13)--(t12);
\path[red edge](t23)--(t22);
\path[red edge](t24)--(t23);
\path[red edge](ti)--(fi);
\path[red edge](f12)--(f13);
\path[red edge](f22)--(f23);
\path[red edge](f13)--(f14);
\path[red edge](f23)--(f24);

\node[svertex](t13) at (5,4){};
\node[svertex](t12) at (6,4){};
\node[svertex](t24) at (4,3){};\node at (4,3.5){$f_2^{2,4}$};
\node[svertex](t23) at (5,3){};
\node[svertex](t22) at (6,3){};
\node[svertex](ti) at (7,3.5){};\node at (7,3.9){$f_2$};
\node[svertex](fi) at (8,3.5){};\node at (8,3.9){$t_2$};
\node[svertex](f12) at (9,4){};
\node[svertex](f13) at (10,4){};
\node[svertex](f14) at (11,4){};\node at (11.5,4.1){$t_2^{1,4}$};
\node[svertex](f22) at (9,3){};
\node[svertex](f23) at (10,3){};
\node[svertex](f24) at (11,3){};\node at (11.5,2.9){$t_2^{2,4}$};
\path[blue edge](ti)--(t12);
\path[blue edge](ti)--(t22);
\path[blue edge](fi)--(f12);
\path[blue edge](fi)--(f22);
\path[red edge](u21)--(t13);
\path[red edge](t13)--(t12);
\path[red edge](t23)--(t22);
\path[red edge](t24)--(t23);
\path[red edge](ti)--(fi);
\path[red edge](f12)--(f13);
\path[red edge](f22)--(f23);
\path[red edge](f13)--(f14);
\path[red edge](f23)--(f24);

\node[svertex](t13) at (5,1){};
\node[svertex](t12) at (6,1){};
\node[svertex](t24) at (4,0){};\node at (4,.5){$t_3^{2,4}$};
\node[svertex](t23) at (5,0){};
\node[svertex](t22) at (6,0){};
\node[svertex](ti) at (7,0.5){};\node at (7,.9){$t_3$};
\node[svertex](fi) at (8,0.5){};\node at (8,.9){$f_3$};
\node[svertex](f12) at (9,1){};
\node[svertex](f13) at (10,1){};
\node[svertex](f14) at (11,1){};\node at (11.5,1.1){$f_3^{1,4}$};
\node[svertex](f22) at (9,0){};
\node[svertex](f23) at (10,0){};
\node[svertex](f24) at (11,0){};\node at (11.5,-.1){$f_3^{2,4}$};
\path[blue edge](ti)--(t12);
\path[blue edge](ti)--(t22);
\path[blue edge](fi)--(f12);
\path[blue edge](fi)--(f22);
\path[red edge](u31)--(t13);
\path[red edge](t13)--(t12);
\path[red edge](t23)--(t22);
\path[red edge](t24)--(t23);
\path[red edge](ti)--(fi);
\path[red edge](f12)--(f13);
\path[red edge](f22)--(f23);
\path[red edge](f13)--(f14);
\path[red edge](f23)--(f24);

\end{tikzpicture}
\caption{A part of the constructed graph for $\ell=4$, $d=1$. Left: the clause gadget $Z_j$ for a clause $c_j=(x_1\lor\lnot x_2\lor x_3)$. Right: the variable gadgets $X_1,X_2$ and $X_3$. Note that~$x_1$~has its second occurrence as a positive literal in $c_j$, $x_2$ has its first occurrence as a negative literal in $c_j$ and $x_3$ has its first occurrence as a positive literal in $c_j$. The clause gadgets where the variables have their other occurrences are not shown.}
\label{fig:example}
\end{figure}
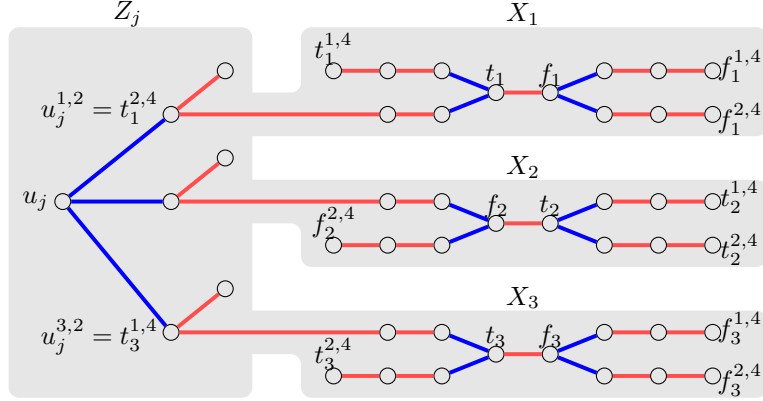

\textit{Intuition}:
Before we prove the correctness of the reduction, we informally describe its idea.

Each variable gadget contains $4\cdot d$ edge-disjoint $c$-colored $P_\ell$s. 
So we have to delete at least $4\cdot d$ edges per variable gadget.
We can make a variable gadget $c$-colored $P_\ell$-free with $4\cdot d$ edge deletions by deleting the edges in $T_i^r$ and~$F_i^b$, or the edges in $F_i^r$ and $T_i^b$.
The former models the assignment~$\mathcal{A}(x_i)=\texttt{true}$, and the latter models the assignment $\mathcal{A}(x_i)=\texttt{false}$.
If we delete $T_i^r$ and $F_i^b$ the vertices~$f_i^{1,z}$ and~$f_i^{2,z}$ are part of a $(c-1)$-colored $P_{\ell-1}$ with no blue edges, while the vertices~$t_i^{1,z}$ and~$t_i^{2,z}$ are part of a $(c-1)$-colored $P_{\ell-2}$.
If we delete $F_i^r$ and $T_i^b$, it is the other way around.
We will be able to make a clause gadget $c$-colored $P_\ell$-free with two edge deletions if and only if there is at least one vertex $u_j^{p,2}$ that is not part of a~$(c-1)$-colored $P_{\ell-1}$ from the connected variable gadget.
Thus, we can make the constructed graph $c$-colored $P_\ell$-free with exactly $k$ edge deletions if and only if each clause is satisfied. 

Before we give the correctness proof, we observe the following.

\begin{Claim}
Let $G':= G- (F_i^b\cup T_i^r)$ and~${G'':= G- (T_i^b\cup F_i^r)}$ for any variable gadget~$X_i$. Then:
\begin{itemize}
\item[\textbf{I}.] No blue edge $e_b\in E_{G'}(X_i)$ is part of a $P_\ell$ in $G'$.
\item[\textbf{II}.] No blue edge $e_b\in E_{G''}(X_i)$ is part of a $P_\ell$ in $G''$.
\end{itemize}
\label{Claim:var_gadget}
 \end{Claim}
\begin{claimproof}
We only prove \textbf{\textit{I}}, since the proof for \textbf{\textit{II}} works analogously.
Let $e_b\in E_{G'}(X_i)$ be a blue edge.
Since $F_i^b\cap E(G')=\emptyset$, we conclude that $e_b\in T_i^b$.
Hence, we know that~${e_b=\{t_i^{q,s},t_i^{q,s+1}\}}$ for some $s$ such that $s\mod (\ell-1)=1$ and $q\in[2]$.
We consider two cases.\\
\textbf{Case 1:} $s=1$. 
By definition of $T_i^r$ we know that $\{t_i^{q,2},t_i^{q,3}\}\in T_i^r$. 
We can conclude that $\deg_{G'}(t_i^{q,2})=1$, since~${T_i^r\cap E(G')=\emptyset}$ .
Furthermore, since~$F_i^b\cap E(G')=\emptyset$, we conclude that $\deg_{G'}(f_i)=1$.
Thus, the longest path in $G'$ containing $e_b$ is the induced subgraph $G'[\{t_i^{q,2},t_i\}\cup W_i\cup\{f_i\}]$.
By construction, it is easy to see that~${|\{t_i^{q,2},t_i\}\cup W_i\cup\{f_i\}|=\ell-1}$.
Thus, $e_b$ is not part of a $P_\ell$ in $G'$.\\
\textbf{Case 2:} $s>1$.
It is easy to see that
\begin{align}
s\mod(\ell-1) = 1 \ 
&\Leftrightarrow\ (s+1)\mod(\ell-1) = 2 \label{eq_1}\\
&\Leftrightarrow\ (s+2-\ell)\mod(\ell-1)=2. \label{eq_2}
\end{align}
From equation (\ref{eq_1}) we conclude that~$\{t_i^{q,s+1},t_i^{q,s+2}\}\in T_i^r$ and from equation (\ref{eq_2}) we conclude that $\{t_i^{q,s+2-\ell},t_i^{q,s+3-\ell}\}\in T_i^r$.
Since $T_i^r\cap E(G')=\emptyset$, we can conclude that~$\deg_{G'}(t_i^{q,s+1})=1$ and $\deg_{G'}(t_i^{q,s+3-\ell})=1$.
Thus, the longest path in $G'$ containing $e_b$ is~$G'[\{t_i^{q,s+3-\ell},\dots,t_i^{q,s+1}\}]$.
It is easy to see that~${|\{t_i^{q,s+3-\ell},\dots,t_i^{q,s+1}\}|=\ell-1}$. Thus, $e_b$ is not part of a $P_\ell$ in $G'$.\\
Hence, no blue edge from $E_{G'}(X_i)$ is part of a $P_\ell$ in $G'$. $\hfill \Diamond$
\end{claimproof}

\textit{Correctness}: 
We show the correctness of the reduction by proving that there is a satisfying assignment for~$\Phi$ if and only if~${(G,k)}$ is a yes-instance of \textsc{$cP_\ell$D}. 

($\Rightarrow$) 
Let $\mathcal{A}:\mathcal{X}\rightarrow\{\texttt{true},\texttt{false}\}$ be a satisfying assignment for~$\Phi$. We will prove that~${(G,k)}$ is a yes-instance of \textsc{$cP_\ell$D} by constructing an edge-deletion set~$S$ of size~$k$ such that $G- S$ is $c$-colored $P_\ell$-free.

For each variable~$x_i\in \mathcal{X}$ we add $4\cdot d$ edges to $S$.
If~$\mathcal{A}(x_i)=\texttt{true}$, then we add~$T_i^r$ and $F_i^b$ to $S$.
If~$\mathcal{A}(x_i)=\texttt{false}$, then we add $F_i^r$ and $T_i^b$ to $S$.
Since~$\mathcal{A}$ satisfies~$\Phi$, there is at least one variable~$x_i\in \mathcal{X}$ such that~$\mathcal{A}(x_i)$ satisfies~$c_j$ for each clause~$c_j\in \mathcal{C}$. 
Let $p\in[3]$ such that the $p$-th literal of $c_j$ satisfies~$c_j$. 
Let $\{\alpha,\beta\}:=[3]\setminus\{p\}$. 
We add~$\{u_j,u_j^{\alpha,2}\}$ and $\{u_j,u_j^{\beta,2}\}$ to~$S$. 
Note that we added exactly two edges per clause. 
Hence, ${|S|=4\cdot d\cdot \eta+2\cdot\mu=k}$.

Next, we show that $G':=G- S$ is $c$-colored $P_\ell$-free.
Since $G$ is a $c$-colored graph, it is sufficient to prove that no blue edge is part of a $c$-colored $P_\ell$ in $G'$.
First, let~$e_b\in E_{G'}(X_i)$ be a blue edge from any variable gadget $X_i$.
Since either~${T_i^r,F_i^b\subseteq S}$ or~${F_i^r,T_i^b\subseteq S}$, we can conclude from Claim \ref{Claim:var_gadget} that $e_b$ is not part of a $P_\ell$ in $G'$.
Note that the proof for Claim \ref{Claim:var_gadget} shows that $e_b$ can neither be part of a $c$-colored~$P_\ell$ with edges from $X_i$, nor be part of a $c$-colored $P_\ell$ with edges from a connected clause gadget $Z_j$.
Next, let $\{u_j,u_j^{p,2}\}\in E_G'(Z_j)$ be the blue edge in $G'$ from any clause gadget~$Z_j$.
By construction of $S$ we know that $\deg_{G'}(u_j)=1$.
This implies that~$\{u_j,u_j^{p,2}\}$ is only part of a $P_{\ell-1}$ in $G'[Z_j]$.
Hence, we can conclude that any $c$-colored $P_\ell$ containing $\{u_j,u_j^{p,2}\}$ has to contain edges from a variable gadget.

We know by construction that~$u_j^{p,2}$ is identified with a vertex from a variable gadget $X_i$ such that $\mathcal{A}(x_i)$ satisfies~$c_j$. 
Without loss of generality, assume~${\mathcal{A}(x_i)=\texttt{true}}$.
Then, $u_j^{p,2}$ is identified with a vertex $t_i^{q,z}$ for some~$q\in[2]$ and~${T_i^r,F_i^b\subseteq S}$.
Assume towards a contradiction that there is a vertex set $V'\subseteq V$ such that $\{u_j,u_j^{p,2}\}\subseteq V'$ and $G'[V']$ is a $c$-colored~$P_\ell$.
Since $\deg_{G'}(u_j)=1$, we conclude that ${\{t_i^{q,z-\ell+2},\dots,t_i^{q,z}\}\subseteq V'}$.
Since~${(z-\ell+2)\mod(\ell-1)=2}$ we know by definition of $T_i^r$ that~${\{t_i^{q,z-\ell+2},t_i^{q,z-\ell+3}\}\in T_i^r}$.
Since~$T_i^r\subseteq S$, we conclude that~${\{t_i^{q,z-\ell+2},t_i^{q,z-\ell+3}\}\in S}$.
Hence, $G'[V']$ is not a $c$-colored~$P_\ell$, a contradiction.
Hence, $\{u_j,u_j^{p,2}\}$ is not part of a $c$-colored~$P_\ell$ in $G'$.
Thus, no blue edge is part of a $c$-colored $P_\ell$ in $G'$.
Hence, $G'$ is $c$-colored $P_\ell$-free.

($\Leftarrow$) 
Let $S$ be an edge-deletion set with $|S|\leq k$ such that $G- S$ is $c$-colored $P_\ell$-free. 
Before we define a satisfying assignment $\mathcal{A}:\mathcal{X}\rightarrow\{\texttt{true}, \texttt{false}\}$ for~$\Phi$, we show two Claims.
First, we show how many edges from each variable gadget and clause gadget have to be in $S$.

\begin{Claim}
$|S \cap E_G(X_i)|= 4\cdot d$ for any variable gadget $X_i$ and $|S\cap E_G(Z_j)\cap E_b|= 2$ for any clause gadget $Z_j$.
\label{Claim:budget}
 \end{Claim}

\begin{claimproof}
First, we will show that $|S\cap E_G(Z_j)|\geq 2$.
Assume towards a contradiction that $|S\cap E_G(Z_j)|< 2$.
First, consider the vertex sets $V_1:=U_j^1\cup \{u_j^{2,2}\}$ and~${V_2:=U_j^1\cup \{u_j^{3,2}\}}$.
Note that $G[V_1]$ and $G[V_2]$ are two different $c$-colored $P_\ell$s and that $V_1\cap V_2=U_j^1$.
Hence, we conclude that $|S\cap U_j^1|=1$.
Next, consider the vertex set $V_3:=U_j^2\cup\{u_j^{3,2}\}$.
The induced subgraph $G[V_3]$ is a $c$-colored $P_\ell$ and $V_3\cap U_j^1=\emptyset$.
This is a contradiction, since $G- S$ is $c$-colored $P_\ell$-free.
Hence,~${|S\cap E_G(Z_j)|\geq 2}$.

By construction, $G[X_i]$ contains $4\cdot d$ edge-disjoint $c$-colored $P_\ell$s.
Hence, we know that~$|S\cap E_G(X_i)|\geq 4\cdot d$.
Since $|S|\leq k= 4\cdot d\cdot \eta + 2\cdot \mu$, we can conclude that~${|S\cap E_G(X_i)|= 4\cdot d}$ and~${|S\cap E_G(Z_j)|= 2}$.

Second, we will show that this implies that~${|S\cap E_G(Z_j)\cap E_b|=2}$.
Assume towards a contradiction that $|S\cap E_G(Z_j)\cap E_b|<2$.
Without loss of generality we can assume that $\{u_j,u_j^{1,2}\},\{u_j,u_j^{2,2}\}\notin S$.
Let~$V_1:=U_j^1\cup \{u_j^{2,2}\}, V_2:=U_j^2\cup \{u_j^{1,2}\}$ and~${V_3:=U_j^3\cup \{u_j^{1,2}\}}$.
The induced subgraphs $G[V_1],G[V_2]$ and $G[V_3]$ are three $c$-colored $P_\ell$s.
This is a contradiction, since $S\cap(E_G(V_\alpha)\cap E_G(V_\beta))=\emptyset$ for~$\alpha,\beta\in[3]$ with $\alpha\neq\beta$ and $|S\cap E_G(Z_j)|=2$.
Hence,~${|S\cap E_G(Z_j)\cap E_b|=2}$. $\hfill \Diamond$
\end{claimproof}

Second, we specify  which edges from variable gadgets have to be deleted.

\begin{Claim} 
For each variable gadget $X_i$ we have ${S\cap E_G(T_i^1\cup T_i^2)=T_i^b}$ \linebreak[4]or ${S\cap E_G(F_i^1\cup F_i^2)=F_i^b}$.
\label{Claim:s}
 \end{Claim}
\begin{claimproof}
First, we show that ${E_G(\{t_i,t_i^{1,2},t_i^{2,2}\})\subseteq S}$ or ${E_G(\{f_i,f_i^{1,2},f_i^{2,2}\})\subseteq S}$.
Consider the four vertex sets in ${\widetilde{X}_i:=\{T_i^1,T_i^2,F_i^1,F_i^2\}}$.
By construction, for each set~${A\in\widetilde{X}_i}$ the induced subgraph $G[A]$ contains $d$ edge-disjoint $c$-colored $P_\ell$s, and for each ${A,B\in\widetilde{X}_i}$ with ${A\neq B}$ we know~${E_G(A)\cap E_G(B)=\emptyset}$. 
From Claim~\ref{Claim:budget} we know that ${|S\cap E_G(X_i)|=4\cdot d}$.
Hence, we conclude that ${|S\cap E_G(A)|=d}$.
Furthermore, we know by construction that ${E_G(A)\cap E_G(W_i\cup\{t_i,f_i\})=\emptyset}$.
Hence, we conclude that ${S\cap E_G(W_i\cup\{t_i,f_i\})=\emptyset}$.
This implies~$E_G(\{t_i,t_i^{1,2},t_i^{2,2}\})\subseteq S$ or~$E_G(\{f_i,f_i^{1,2},f_i^{2,2}\})\subseteq S$, since the induced subgraphs~$G[W_i\cup\{t_i,f_i,t_i^{q,2},f_i^{q',2}\}]$ are $c$-colored $P_\ell$s for each ${q,q'\in[2]}$.

Without loss of generality, assume~${E_G(\{t_i,t_i^{1,2},t_i^{2,2}\})\subseteq S}$.
We can conclude that~$|S\cap (T_i^1\setminus\{t_i\})|=d-1=|S\cap (T_i^2\setminus\{t_i\})|$, since $|S\cap T_i^1|=d=|S\cap T_i^2|$, as shown above.
Furthermore, the induced subgraphs~${G[T_i^1\setminus\{t_i\}]}$ and~${G[T_i^2\setminus\{t_i\}]}$ are paths of $d\cdot(\ell-1)$ vertices where any $\ell$ consecutive vertices form a $c$-colored~$P_\ell$.
Thus, from Lemma \ref{lem:path} we conclude ${T_i^b\subseteq S}$.
Since $|T_i^b|=2\cdot d$ and the induced subgraph~${G[F_i^1\cup F_i^2]}$ contains~${2\cdot d}$ edge-disjoint $c$-colored $P_\ell$s, we conclude that~${S\cap E_G(T_i^1\cup T_i^2)=T_i^b}$ from Claim \ref{Claim:budget}.
Thus,~${S\cap E_G(T_i^1\cup T_i^2)=T_i^b}$ or~${S\cap E_G(F_i^1\cup F_i^2)=F_i^b}$. $\hfill \Diamond$
\end{claimproof}

We show how to construct an equivalent solution $S'$ such that either~$S'\cap E_G(X_i)=F_i^b\cup T_i^r$ or~${S'\cap E_G(X_i)=T_i^b\cup F_i^r}$ for each variable gadget~$X_i$.

If ${S\cap E_G(T_i^1\cup T_i^2)=T_i^b}$, then we can conclude that ${S\cap E_G(X_i)=T_i^b\cup F}$ such that~$F\subseteq E_G(F_i^1\cup F_i^2)$ from Claim \ref{Claim:s}.
We will show that ${S':=(S\setminus F)\cup F_i^r}$ is an equivalent solution.
Since~${|T_i^b|=2\cdot d}$, we can conclude that~${|F|=2\cdot d}$ from Claim~\ref{Claim:budget}.
Hence, it is easy to see that~${|S'|=|S|}$.
Now, we show that~${G':=G- S'}$ is $c$-colored $P_\ell$-free.
Since $G$ is constructed from a \textsc{(3,B2)-SAT} formula, we know that for each clause gadget $Z_j$ and each~${p\in[3]}$, the vertex $u_j^{p,2}$ is identified with exactly one vertex from a variable gadget.
Furthermore, we know that the three edges ${\{u_j,u_j^{p,2}\}}$ are the only blue edges in $E_G(Z_j)$ and that $G$ is a $c$-colored graph.
Hence, it is sufficient to show that for any variable gadget~$X_i$ there is no blue edge~${e_b\in E_{G'}(X_i)}$ that is part of a $c$-colored $P_\ell$ in $G'$ and for each clause gadget with~${u_j^{p,2}=t_i^{q,z}}$ or~$u_j^{p,2}=f_i^{q,z}$ for some~$p\in[3],q\in[2]$, the blue edge~${\{u_j,u_j^{p,2}\}}$ is not part of a $c$-colored $P_\ell$ in~$G'$.
Since~${S'\cap E_G(X_i)=F_i^b\cup T_i^r}$, it follows by Claim \ref{Claim:var_gadget} that no blue edge from~$E_{G'}(X_i)$ is part of a $c$-colored~$P_\ell$ in $G'$.
Let $Z_j$ be a clause gadget such that ${u_j^{p,2}=t_i^{q,z}}$ or~$u_j^{p,2}=f_i^{q,z}$ for some~${p\in[3],q\in[2]}$.
From Claim \ref{Claim:budget} we conclude that $S'$ contains exactly two blue edges from $E_G(Z_j)$.
Hence, $\deg_{G'}(u_j)=1$.
We consider two cases.\\
\textbf{Case 1:} $u_j^{p,2}=t_i^{q,z}$ for some $p\in[3],q\in[2]$.
We will show that ${\{u_j,u_j^{p,2}\}\in S'}$.
Assume towards a contradiction that $\{u_j,u_j^{p,2}\}\notin S'$.
Since $S'\cap E_G(Z_j)=S\cap E_G(Z_j)$, we conclude that ${\{u_j,u_j^{p,2}\}\notin S}$.
And since ${S\cap E_G(T_i^1\cup T_i^2)=T_i^b}$, we conclude that the induced subgraph ${G[\{u_j,t_i^{q,z},\dots,t_i^{q,z-\ell+2}\}]-S}$ is a $c$-colored $P_\ell$.
This is a contradiction, since $G- S$ is $c$-colored $P_\ell$ free.
Hence, ${\{u_j,u_j^{p,2}\}\in S'}$ and thus, ${\{u_j,u_j^{p,2}\}}$ is not part of a $c$-colored $P_\ell$ in $G'$.\\
\textbf{Case 2:} $u_j^{p,2}=f_i^{q,z}$ for some $p\in[3],q\in[2]$.
Since it is obvious that ${\{u_j,u_j^{p,2}\}}$ is not part of a $c$-colored $P_\ell$ in $G'$ if~${\{u_j,u_j^{p,2}\}\in S'}$, we only consider the case in which~${\{u_j,u_j^{p,2}\}\notin S'}$.
Since ${(z-\ell+2)\mod(\ell-1)=2}$, we know by definition of $F_i^r$ that~${\{f_i^{q,z-\ell+2},f_i^{q,z-\ell+3}\}\in F_i^r}$.
This implies that~${\deg_{G'}(f_i^{q,z-\ell+3})=1}$, since $F_i^r\subseteq S'$.
Thus, the longest path in $G'$ that contains~${\{u_j,u_j^{p,2}\}}$ is the induced subgraph ${G'[V':=\{u_j,f_i^{q,z},\dots,f_i^{q,z-\ell+3}\}]}$.
It is not hard to see that~${|V'|=\ell-1}$.
Thus,~$\{u_j,u_j^{p,2}\}$ is not part of a $c$-colored $P_\ell$ in $G'$.

Hence, no blue edge is part of a $c$-colored $P_\ell$ in $G'$ and therefor, $S'$ is a solution such that~$S'\cap E_G(X_i)=T_i^b\cup F_i^r$.

If ${S\cap E_G(T_i^1\cup T_i^2)\neq T_i^b}$, then we conclude that ${S\cap E_G(F_i^1\cup F_i^2)=F_i^b}$ from Claim~\ref{Claim:s}. Hence, ${S\cap E_G(X_i)=F_i^b\cup T}$ such that~${T\subseteq E_G(T_i^1\cup T_i^2)}$.
With an analogous argument, we can show that ${S':=(S\setminus T)\cup T_i^r}$ is an equivalent solution.
Thus, we have shown how to construct a solution $S'$ from~$S$ such that either~${S'\cap E_G(X_i)=F_i^b\cup T_i^r}$ or ${S'\cap E_G(X_i)=T_i^b\cup F_i^r}$

Now we can define a satisfying assignment ${\mathcal{A}:\mathcal{X}\rightarrow\{\texttt{true}, \texttt{false}\}}$ for~$\Phi$ as:
\begin{center}
$\mathcal{A}(x_i):=
\begin{cases} 
\texttt{true}&\text{if }S'\cap E_G(X_i)=F_i^b\cup T_i^r\\
\texttt{false}&\text{if }S'\cap E_G(X_i)=T_i^b\cup F_i^r.
\end{cases}$
\end{center}

It remains to show that $\mathcal{A}$ satisfies $\Phi$. 
Let $c_j\in\mathcal{C}$. 
We know that there is exactly one $p\in[3]$ such that $\{u_j, u_j^{p,2}\}\notin S$. 
Let ${x_i\in\mathcal{X}}$ be the variable that occurs as the $p$-th literal in $c_j$.
If the $p$-th literal in $c_j$ is a positive literal, we know that the vertex $u_j^{p,2}$ is identified with the vertex $t_i^{q,z}$ from the variable gadget $X_i$ for~${q\in[2]}$. 
Since $S'$ is a solution and ${\{u_j, u_j^{p,2}\}\notin S'}$, we conclude that ${T_i^r\subseteq S'}$.
Hence,~${S'\cap E_G(X_i)=F_i^b\cup T_i^r}$.
Thus, ${\mathcal{A}(x_i)=\texttt{true}}$ and therefore ${\mathcal{A}(x_i)}$ satisfies $c_j$.
If the $p$-th literal in $c_j$ is a negative literal, the argument works analogously.
Hence, each clause $c_j$ is satisfied by $\mathcal{A}$ and thus $\Phi$ is satisfied as well.
\end{proof}

If the girth of a graph $G$ is greater than $\ell$, then each subgraph of~$G$ that is isomorphic to a $c$-colored~$P_\ell$ is an induced subgraph.
This implies the following.

\begin{corollary}
\textsc{$cP_\ell$D} is NP-hard for each $\ell\geq 4$ and $c\in [2,\ell-2]$ on strictly non-cascading graphs.
\end{corollary}

Furthermore, consider the problem~\textsc{All-$cP_\ell$D} where the input consists of a graph~$G$ and an integer~$k$ and the task is to destroy all (not necessarily induced)~$c$-colored~$P_\ell$s with at most~$k$ edge deletions. Theorem~\ref{thm:cPD} then implies the following.

\begin{corollary}
~\textsc{All-$cP_\ell$D} is NP-hard for each $\ell\geq 4$ and $c\in [2,\ell-2]$.
\end{corollary}

Next, we show that also the remaining case where $c=\ell-1$ is NP-hard. 

\begin{theorem}
\textsc{$(\ell-1)P_\ell$D} is NP-hard for any $\ell\geq 4$ even if the maximum degree of~$G$ is $16$. 
\label{thm:lPD}
\end{theorem}

\begin{proof}
We prove this theorem by giving a polynomial time reduction from the NP-hard \textsc{2$P_3$D} problem~\cite{gruettemeier19}. 

\textit{Construction}:
Let $(G,k)$ be an instance of \textsc{2$P_3$D}.
We will show how to construct an equivalent instance~$(H,k)$ of \textsc{$(\ell-1)P_\ell$D} for any $\ell\geq 4$.
We use the instance~$(G,k)$ and add vertices and edges with new colors.
Hence,~${V(G)\subseteq V(H)}$ and~${E(G)\subseteq E(H)}$.
For each vertex $v\in V(G)$ we add $\deg_{G}(v)\cdot (\ell-3)$ new vertices~$v_i^j$ for $i\in[\deg_G(v)]$ and~$j\in[\ell-3]$ to $V(H)$.
Recall that yellow is the third color.
We add yellow edges $\{v,v_i^1\}$ for each~${i\in[\deg_G(v)]}$.
And for each $j\in[\ell-4]$ we add an edge~${\{v_i^j,v_i^{j+1}\}}$ with color~${j+3}$.
Note that for each~${i\in[\deg_G(v)]}$ the induced subgraph $H[\{v,v_i^1,\dots,v_i^{\ell-3}\}]$ is an~${(\ell-3)}$-colored~$P_{\ell-2}$.
Hence, each vertex~${v\in V(G)}$~is part of $\deg_G(v)$ edge-disjoint~$(\ell-3)$-colored~$P_{\ell-2}$ in~$H$ (see Fig.~\ref{fig:zcPz+1D_NP-hard}). 
The budget~$k$ remains the same.

By construction, we know that $\deg_{H}(v) = 2\cdot \deg_G(v)$ for each $v\in V(G)$, and that~$\deg_{H}(v_i^j)\in[2]$ for the new vertices $v_i^j\in V(H)\setminus V(G)$. 
Since \textsc{2$P_3$D} is NP-hard even if the maximum degree of $G$ is eight~\cite{gruettemeier19}, the correctness of the reduction will imply the NP-hardness even if the maximum degree of~$H$ is~16.

\textit{Correctness}:
We will now prove the correctness of the reduction by showing that~$(G,k)$ is a yes-instance of \textsc{2$P_3$D} if and only if $(H,k)$ is a yes-instance \linebreak[4]of~\textsc{${(\ell-1)P_\ell}$D}.

($\Rightarrow$)
Let $S$ be an edge-deletion set of size at most $k$ such that $G- S$ is bicolored~$P_3$-free. 
Since~$H$ is a $(\ell-1)$-colored graph, each $(\ell-1)$-colored~$P_{\ell}$ in $H$ has to include exactly one red edge and exactly one blue edge. 
By construction, we know that an edge $e\in E(H)$ is red or blue if and only if~$e\in E(G)$. 
Since $v_i^1$ is the only vertex from $\{v_i^1,\dots,v_i^{\ell-3}\}$ that is adjacent to a vertex from $V(G)$ for each $i\in[\deg_G(v)]$, we conclude that each~$(\ell-1)$-colored $P_{\ell}$ in $H$ has to include an induced bicolored~$P_3$ from~$G$. 
Since~${G- S}$ is bicolored $P_3$-free,~$H- S$ is $(\ell-1)$-colored $P_{\ell}$-free. 
Thus, $S$ is a solution for~$(H,k)$.

($\Leftarrow$)
Let $S$ be an edge-deletion set of size at most $k$ such that $H- S$ is $(\ell-1)$-colored $P_{\ell}$-free.

First, we will show how to construct an equivalent solution $S'$ with $|S'|\leq |S|$ such that $S'$ only contains blue, red and yellow edges.
Note that this is only necessary if~$\ell>4$.
Let $\{v_i^j,v_i^{j+1}\}\in S$ be an edge with color $c>3$.
Let~${V'\subseteq V(H)}$ be a vertex set such that $v_i^j,v_i^{j+1}\in V'$ and $H[V']$ is an $(\ell-1)$-colored~$P_\ell$.
Since~$H$ is an $(\ell-1)$-colored graph, each $(\ell-1)$-colored $P_\ell$ has to include a red edge and a blue edge.
By construction, the only vertices, to which blue and red edges can be incident, are the vertices from $V(G)$.
Hence, we conclude that $v\in V'$. 
Thus,~${S':=(S\setminus\{v_i^j,vi^{j+1}\})\cup \{\{v,v_i^1\}\}}$ is an equivalent solution that only contains blue, red, and yellow edges.

\begin{figure}[t!]
\centering
\subfloat{
\centering
\begin{tikzpicture}
\node[vertex](u) at (2,0){};\node at(2,-.4){$x$};
\node[vertex](v) at (4,0){};\node at(4,-.4){$y$};
\node[vertex](w) at (3,2){};\node at(3,2.5){$w$};
\node[vertex](x) at (2,4){};\node at(2,4.4){$u$};
\node[vertex](y) at (4,4){};\node at(4,4.4){$v$};
\path[red edge](u)--(v);
\path[red edge](u)--(w);
\path[red edge](v)--(w);
\path[blue edge](w)--(x);
\path[blue edge](w)--(y);
\end{tikzpicture}
}
\subfloat{
\centering
\begin{tikzpicture}
\node[vertex](u) at (2,0){};\node at(2,-.4){$x$};
\node[vertex](v) at (4,0){};\node at(4,-.4){$y$};
\node[vertex](w) at (3,2){};\node at(3,2.5){$w$};
\node[vertex](x) at (2,4){};\node at(2,4.4){$u$};
\node[vertex](y) at (4,4){};\node at(4,4.4){$v$};
\path[red edge](u)--(v);
\path[red edge](u)--(w);
\path[red edge](v)--(w);
\path[blue edge](w)--(x);
\path[blue edge](w)--(y);

\node[svertex](11) at (5.5,0){};\node at (5.5,.4){$y_2^1$};
\node[svertex](1z) at (7,0){};\node at (7.6,.2){$y_2^{c-2}$};
\path[yellow edge](v)--(11);
\path[black dots](11)--(1z);
\node[svertex](11) at (5.5,1){};\node at (5.5,1.4){$y_1^1$};
\node[svertex](1z) at (7,1){};\node at (7.6,1.2){$y_1^{c-2}$};
\path[yellow edge](v)--(11);
\path[black dots](11)--(1z);
\node[svertex](11) at (5.5,2){};\node at (5.5,2.4){$w_4^1$};
\node[svertex](1z) at (7,2){};\node at (7.6,2.2){$w_4^{c-2}$};
\path[yellow edge](w)--(11);
\path[black dots](11)--(1z);
\node[svertex](11) at (5.5,3){};\node at (5.5,3.4){$w_3^1$};
\node[svertex](1z) at (7,3){};\node at (7.6,3.2){$w_3^{c-2}$};
\path[yellow edge](w)--(11);
\path[black dots](11)--(1z);
\node[svertex](11) at (5.5,4){};\node at (5.5,4.4){$v_1^1$};
\node[svertex](1z) at (7,4){};\node at (7,4.4){$v_1^{c-2}$};
\path[yellow edge](y)--(11);
\path[black dots](11)--(1z);

\node[svertex](11) at (.5,0){};\node at (.5,.4){$x_2^1$};
\node[svertex](1z) at (-1,0){};\node at (-1.3,.2){$x_2^{c-2}$};
\path[yellow edge](u)--(11);
\path[black dots](11)--(1z);
\node[svertex](11) at (0.5,1){};\node at (.5,1.4){$x_1^1$};
\node[svertex](1z) at (-1,1){};\node at (-1.3,1.2){$x_1^{c-2}$};
\path[yellow edge](u)--(11);
\path[black dots](11)--(1z);
\node[svertex](11) at (0.5,2){};\node at (.5,2.4){$w_2^1$};
\node[svertex](1z) at (-1,2){};\node at (-1.3,2.2){$w_2^{c-2}$};
\path[yellow edge](w)--(11);
\path[black dots](11)--(1z);
\node[svertex](11) at (0.5,3){};\node at (.5,3.4){$w_1^1$};
\node[svertex](1z) at (-1,3){};\node at (-1.3,3.2){$w_1^{c-2}$};
\path[yellow edge](w)--(11);
\path[black dots](11)--(1z);
\node[svertex](11) at (0.5,4){};\node at (.5,4.4){$u_1^1$};
\node[svertex](1z) at (-1,4){};\node at (-1,4.4){$u_1^{c-2}$};
\path[yellow edge](x)--(11);
\path[black dots](11)--(1z);
\end{tikzpicture}
}
\caption{(a) A bicolored graph $G$. (b) The graph $H$ with the new vertices and edges. The black dotted lines each represent a $(\ell-4)$-colored path containing $\ell-3$ edges that are neither blue, red, nor yellow.}
\label{fig:zcPz+1D_NP-hard}
\end{figure}
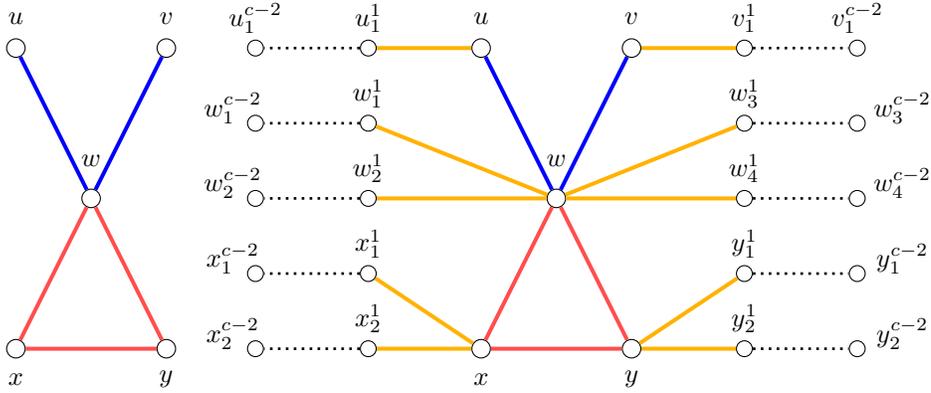

Second, we will show how to construct an equivalent solution $S''$ from $S'$ that only contains blue and red edges.
Recall that $E_{H}(\{v\},\{v_i^1\ |\ i\in[\deg_G(v)]\})$ denotes the set of all yellow edges that are incident to a vertex $v\in V(G)$.
Let $\{v,v_i^1\}\in S'$ be a yellow edge.
We have to consider two cases.\\
\textbf{Case 1:} $E_{H}(\{v\},\{v_i^1\ |\ i\in[\deg_G(v)]\})\subseteq S'$. 
Since $G$ is a $(\ell-1)$-colored graph, we know by construction that each $(\ell-1)$-colored $P_\ell$ that contains the vertex $v$, has to contain a red or blue edge that is incident to $v$. 
Thus, we can construct an equivalent solution by swapping the $\deg_G(v)$~many yellow edges that are incident to $v$ with the $\deg_G(v)$ blue or red edges that are incident to $v$.

We set~$S'':=(S'\setminus E_{H}(\{v\},\{v_i^1\ |\ i\in[\deg_G(v)]\}))\cup E_G(\{v\},N_G(v))$.
Since there are no red or blue edges in $H- S''$ that are incident to $v$, we conclude that $v$ is not part of a $(\ell-1)$-colored $P_\ell$ in~$H- S''$ and therefore no edge in~$E_{H}(\{v\},\{v_i^1\ |\ i\in[\deg_G(v)]\})$ is part of a $(\ell-1)$-colored $P_\ell$ in~$H- S''$.
Hence,~$S''$ is an equivalent solution with no yellow edges.\\
\textbf{Case 2:} $E_{H}(\{v\},\{v_i^1\ |\ i\in[\deg_G(v)]\})\nsubseteq S'$.
Then, for some $\alpha\in[\deg_G(v)]$ there is a yellow edge~${\{v,v_\alpha^1\}\notin S'}$.
Let $\{v,v_\beta^1\}\in S'$ such that $\beta\neq\alpha$.
Since $S'$ is a solution and $H[\{v,v_\alpha^1,\dots,v_\alpha^{c-2}\}]\cong H[\{v,v_\beta^1,\dots,v_\beta^{c-2}\}]$, we conclude that $S'\setminus \{v,v_\beta^1\}$ is a solution.
We set $S'':=(S'\setminus E_{H}(\{v\},\{v_i^1\ |\ i\in[\deg_G(v)]\}$ and get an equivalent solution with no yellow edges.

Thus, $S''$ only contains red and blue edges.
It remains to show that $S''$ is a solution for $(G,k)$. 
Assume towards a contradiction that there is an induced bicolored $P_3$ in $G- S''$. 
Without loss of generality, this implies that there is an edge set~${\{\{u,v\}\{v,w\}\}\subseteq (E(G)\setminus S'')}$~so that~$\{u,v\}$~is a blue edge, $\{v,w\}$ is a red edge and~${\{u,w\}\notin (E(G)\setminus S'')}$.
Since $S''$ only contains blue or red edges, we can conclude that there is a vertex set $U_\alpha:=\{u,u_\alpha^1,\dots,u_\alpha^{\ell-3}\}$ such that $H[U_\alpha]- S''$ is an induced $(\ell-3)$-colored $P_{\ell-2}$ with no blue or red edges. 
Hence, we conclude that~$H[U_\alpha\cup\{v,w\}]$ is an induced $(\ell-1)$-colored $P_\ell$.
This is a contradiction, since~$S''$ is an solution for $(H,k)$.
Hence, $G- S''$ is bicolored~$P_3$-free.
Thus, $S''$ is a solution for $(G,k)$.
\end{proof}

Since \textsc{2$P_3$D} is NP-hard~\cite{gruettemeier19} and \textsc{1$P_\ell$D} is NP-hard~\cite{Mallah88}, we can conclude the following from Theorem~\ref{thm:cPD} and Theorem~\ref{thm:lPD}.

\begin{corollary}
\textsc{$cP_\ell$D} is NP-hard for each $\ell\geq 3$ and each $c\in[\ell-1]$.
\end{corollary}

\subsection{$c$-colored $C_\ell$ Deletion}
\label{sec:ccd}
Next, we analyze the complexity of \textsc{$cC_\ell$D} which is known to be NP-hard for any $\ell\geq 3$ and $c=1$~\cite{Yannakakis81}. 
We show that \textsc{$cC_\ell$D} is NP-hard for any~${\ell\geq 3}$ and~$c\in[\ell]$.
The result is based on a reduction from the NP-hard problem \textsc{Vertex Cover (VC)}. In \textsc{VC} one is given a graph~$G$ and an integer~$k$ and asks if there is a subset~$S$ of vertices of~$G$ such that~$|S| \leq k$ and every edge of~$G$ has at least one endpoint in~$S$. 
But before we consider \textsc{$cC_\ell$D}, we establish an NP-hardness result for \textsc{VC} on $C_3$-free and tripartite graphs that we will use in our reduction.

\begin{lemma}
\textsc{VC} is NP-hard even if $G$ is $C_3$-free and tripartite.
\label{lem:vc}
\end{lemma}

\begin{proof}
\textsc{Independent Set (IS)} is NP-hard on 2-subdivision graphs~\cite{poljak74}.
Since $(G,k)$ is a yes-instance of \textsc{IS} if and only if $(G,n-k)$ is a yes-instance of \textsc{VC}~\cite{garey79}, we conclude that~\textsc{VC} is NP-hard on~2-subdivision graphs.
Since each 2-subdivision graph is~$C_3$-free and tripartite, \textsc{VC} is NP-hard on $C_3$-free and tripartite graphs.
\end{proof}

Now we can show the NP-hardness of \textsc{$cC_\ell$D} for $c=\ell$.
With this result we will be able to prove the NP-hardness for all $c\in[\ell]$.

\begin{lemma}
\textsc{$\ell C_\ell$D} is NP-hard for any $\ell\geq 3$ even if $G$ has girth $\ell$ and every~$C_\ell$ in $G$ is $\ell$-colored.
\label{thm:cd}
\end{lemma}

\begin{proof}
We give a polynomial time reduction from the NP-hard \textsc{VC} problem on $C_3$-free and tripartite graphs (see Lemma \ref{lem:vc}).
Note that this reduction is very similar to the one given by Yannakakis~\cite{Yannakakis81} to prove the NP-hardness for~$c=1$.

Let $(H,k)$ be an instance of \textsc{VC} such that $H$ is tripartite and $C_3$-free.
Before we describe how to construct an equivalent instance $(G,k)$ of~\textsc{$\ell C_\ell$D}, we define two functions to color the vertices and edges of $H$.
Since $H$ is tripartite, there is a function~${\varphi:V(H)\rightarrow [3]}$ such that $\varphi(u)\neq\varphi(v)$ for each edge~${\{u,v\}\in E(H)}$.
Hence, there is exactly one ${\gamma\in[3]\setminus \{\varphi(u),\varphi(v)\}}$. Thus, $\psi:E(H)\rightarrow[3]$ with~$\psi(\{u,v\})=\gamma$ is a well defined function.

\textit{Construction}:
Now we show how to construct $(G,k)$.
If $\ell=3$, then we assign the color $\psi(\{u,v\})$ to each edge $\{u,v\}\in E(H)$. 
If $\ell>3$, then we subdivide each edge $\{u,v\}\in E(H)$ with vertices~$W^{uv}:=\{w_1^{uv},\dots,w_{\ell-3}^{uv}\}$.
We then color the edges as follows. 
We assign color $\psi(\{u,v\})$ to the edge $\{u,w_1^{uv}\}$ and color $\ell$ to the edge~$\{w_{\ell-3}^{uv},v\}$.
For $i\in\{1,\dots,\ell-4\}$, the edges $\{w_i^{uv},w_{i+1}^{uv}\}$ are assigned color $i+3$.
Hence, the induced subgraph $G[\{u,v\}\cup W^{uv}]$ is an $(\ell-2)$-colored $P_{\ell-1}$ with all colors except $\varphi(u)$ and~$\varphi(v)$.
Then, we add a vertex $\alpha$.
To complete the reduction we add an edge $\{v,\alpha\}$ with color $\varphi(v)$ for each $v\in V(H)$.
The budget $k$ remains the same.

Before we show the correctness of the reduction we prove the following Claim about the structure of~$G$.
\begin{Claim}
The girth of $G$ is $\ell$ and every $C_\ell$ in $G$ is $\ell$-colored.
 \end{Claim}
\begin{claimproof}
Obviously $(G,k)$ is a trivial yes-instance if the girth of $G$ is greater than~$\ell$.
Since $H$ is $C_3$-free and every edge in $E(H)$ corresponds to a induced $P_{\ell-1}$ in $G$, we conclude that the girth of $G[V\setminus\{\alpha\}]$ is greater than $\ell$.
Furthermore, we know by construction that for each edge $\{u,v\}\in E(H)$ the induced subgraph~${G[\{\alpha,u,v\}\cup W^{uv}]}$ is an $\ell$-colored $C_\ell$, since $u$ and $v$ are connected by an $(\ell-2)$-colored $P_{\ell-1}$ with all colors except $\varphi(u)$ and $\varphi(v)$, the edge $\{u,\alpha\}$ has color $\varphi(u)$ and the edge $\{v,\alpha\}$ has color $\varphi(v)$.
Hence, we conclude that the girth of $G$ is $\ell$ and every $C_\ell$ in $G$ is~$\ell$-colored. $\hfill \Diamond$
\end{claimproof}


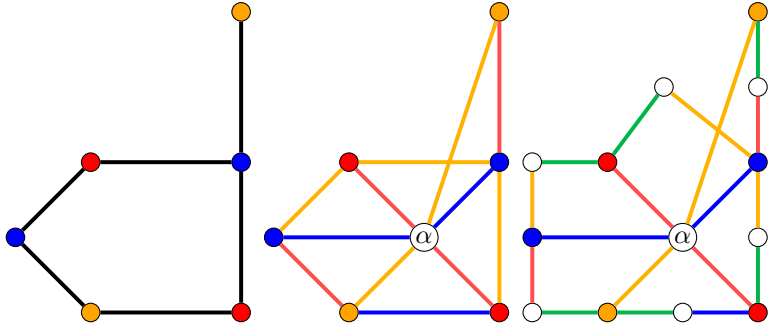
\begin{figure}[t!]
\centering
\subfloat{
\centering
\begin{tikzpicture}
\node[rvertex](u) at(0,2){};
\node[bvertex](v) at(2,2){};
\node[rvertex](w) at(2,0){};
\node[yvertex](x) at(0,0){};
\node[yvertex](y) at(2,4){};
\node[bvertex](z) at(-1,1){};
\path[black edge](v)--(u);
\path[black edge](v)--(w);
\path[black edge](x)--(z);
\path[black edge](x)--(w);
\path[black edge](v)--(y);
\path[black edge](z)--(u);

\end{tikzpicture}

}
\subfloat{
\centering
\begin{tikzpicture}
\node[rvertex](u) at(0,2){};
\node[bvertex](v) at(2,2){};
\node[rvertex](w) at(2,0){};
\node[yvertex](x) at(0,0){};
\node[yvertex](y) at(2,4){};
\node[bvertex](z) at(-1,1){};
\node[vertex](a)at(1,1){$\alpha$};
\path[yellow edge](v)--(u);
\path[yellow edge](v)--(w);
\path[red edge](x)--(z);
\path[yellow edge](u)--(z);
\path[blue edge](x)--(w);
\path[red edge](v)--(y);
\path[red edge](a)--(u);
\path[red edge](a)--(w);
\path[yellow edge](a)--(x);
\path[yellow edge](a)--(y);
\path[blue edge](a)--(v);
\path[blue edge](a)--(z);
\end{tikzpicture}

}
\subfloat{
\centering
\begin{tikzpicture}
\node[rvertex](u) at(0,2){};
\node[bvertex](v) at(2,2){};
\node[rvertex](w) at(2,0){};
\node[yvertex](x) at(0,0){};
\node[yvertex](y) at(2,4){};
\node[bvertex](z) at(-1,1){};
\node[vertex](a)at(1,1){$\alpha$};
\node[vertex](b)at(-1,0){};
\node[vertex](c)at(-1,2){};
\node[vertex](d)at(1,0){};
\node[vertex](e)at(2,1){};
\node[vertex](f)at(.75,3){};
\node[vertex](g)at(2,3){};
\path[yellow edge](v)--(f);
\path[green edge](f)--(u);
\path[yellow edge](v)--(e);
\path[green edge](e)--(w);
\path[red edge](b)--(z);
\path[green edge](b)--(x);
\path[yellow edge](c)--(z);
\path[green edge](c)--(u);
\path[blue edge](w)--(d);
\path[green edge](d)--(x);
\path[red edge](v)--(g);
\path[green edge](g)--(y);
\path[red edge](a)--(u);
\path[red edge](a)--(w);
\path[yellow edge](a)--(x);
\path[yellow edge](a)--(y);
\path[blue edge](a)--(v);
\path[blue edge](a)--(z);
\end{tikzpicture}
}
\caption{(a) A tripartite, $C_3$-free graph $H$. (b) The graph constructed on input (a) if $\ell=3$. (c) The graph constructed on input (a) if $\ell=4$.}
\end{figure}

\textit{Correctness}:
To show the correctness of the reduction, we prove that $(H,k)$ is a yes-instance of \textsc{VC} if and only if $(G,k)$ is a yes-instance of \textsc{$\ell C_\ell$D}.

($\Rightarrow$)
Let $V'\subseteq V(H)$ be a vertex cover of size at most $k$. 
Consider the edge-deletion set $S:=\{\{v,\alpha\}\ |\ v\in V'\}$.
Since $|V'|=|S|$, we know that $|S|\leq k$.
To show that $G'=G- S$ is $C_\ell$-free, it is sufficient to show that $\alpha$ is not part of a $C_\ell$ in $G'$.
Let $\{\alpha,v\}\in E(G')$. 
By construction, we know that $v\in V(H)$ and~${v\notin V'}$.
Assume towards a contradiction that $\{v,\alpha\}$ is part of a $C_\ell$.
Then, there is an edge~${\{u,\alpha\}\in E(G')}$ such that $u\neq v$ and $u$ is connected to $v$ by an induced $P_{\ell-1}$. 
This implies that $u\in V(H)$ and $u\notin V'$.
Since $u,v\in V(H)$ and $u$ is connected to $v$ by induced $P_{\ell-1}$ in $G$, we can conclude that $\{u,v\}\in E(H)$.
But that is a contradiction, since $V'$ is a vertex cover of $H$ and~$v,u\notin V'$.
Hence, $G'$ is $C_\ell$-free.

($\Leftarrow$)
Let $S$ be an edge-deletion set of size at most $k$ such that $G':=G- S$ is $\ell$-colored $C_\ell$-free.
Let~$E':=\{\{u,v\}\in E(G)\ |\ u,v\neq\alpha\}$ be the set of edges from $G$ that are not incident with $\alpha$.
We observe that each edge $e\in E'$ is part of at most one $C_\ell$ in $G$.
Since $\alpha$ is part of every $C_\ell$ in $G$, we know that there is a vertex~${\beta\in V(G)}$ such that $S'=(S\setminus\{e\})\cup\{\{\alpha,\beta\}\}$ is an equivalent solution.
Since~$\beta$ is adjacent to $\alpha$, we can conclude that $\beta\in V(H)$.

To finish the proof, we show that $V':=\{\beta\ |\ \{\alpha,\beta\}\in S'\}$ is a vertex cover for $H$.
Let $\{v_1,v_2\}\in E(H)$ be an edge from $H$. 
Assume towards a contradiction that $v_1,v_2\notin V'$.
This implies that $\{\alpha,v_1\},\{\alpha,v_2\}\notin S'$.
By construction, we know that $v_1$ and $v_2$ are connected by an induced $P_{\ell-1}$ that consists of edges from $E'$ and contains all colors except $\varphi(v_1)$ and~$\varphi(v_2)$.
Since $\{\alpha,v_1\}$ has color~$\varphi(v_1)$,~$\{\alpha,v_2\}$ has color~$\varphi(v_2)$, and $S'\cap E'=\emptyset$, we conclude that $v_1$ and $v_2$ are part of a $\ell$-colored~$C_\ell$ in $G'$. 
That is a contradiction, since $G'$ is $\ell$-colored $C_\ell$-free.
So $v_1\in V'$ or $v_2\in V'$. 
thus,~$V'$ is a vertex cover for~$H$.
\end{proof}

Now we will use this result to prove the NP-hardness of \textsc{$cC_\ell$D} for all $c\in[\ell-1]$.
\begin{lemma} 
\textsc{$cC_\ell$D} is NP-hard for each $c\in[\ell-1]$ even if $G$ has girth~$\ell$ and each $C_\ell$ in $G$ is $c$-colored.
\label{lem:cd}
\end{lemma}
\begin{proof}
We give a reduction from \textsc{$\ell C_\ell$D} on graphs where each $C_\ell$ is~$\ell$-colored.
Let $(H,k)$ be an instance of \textsc{$\ell C_\ell$D} where the girth of $H$ is $\ell$ and each $C_\ell$ in $H$ is $\ell$-colored.
To construct an equivalent instance~$(G,k)$ of~\textsc{$cC_\ell$D} we recolor the edges.
For each $\alpha\in[c-1]$ we set~${E_\alpha(G):=E_\alpha(H)}$. 
Next, we set~${E_c(G):=E_c(H)\cup\ldots\cup E_\ell(H)}$.
Note that the vertex set and the budget $k$ remains the same.
Since we do not add new edges, the girth remains the same, and since each $C_\ell$ in $H$ is $\ell$-colored, we know by construction that each $C_\ell$ in $G$ is~$c$-colored.
Since every $C_\ell$ in $H$ is $\ell$-colored, for each vertex set $V'\subseteq V(H)$ the induced subgraph $H[V']$ is an $\ell$-colored $C_\ell$ if and only if the induced subgraph~$G[V']$~is a $c$-colored $C_\ell$. 
Hence, $(H,k)$ is a yes-instance of \textsc{$\ell C_\ell$D} if and only if~$(G,k)$ is a yes-instance of \textsc{$cC_\ell$D}.
\end{proof}

If $G$ has girth~$\ell$, then every $c$-colored $C_\ell$ in~$G$ is an induced subgraph of~$G$.
Thus, we may conclude the following from Lemmas~\ref{thm:cd} and~\ref{lem:cd}.

\begin{theorem}
\textsc{$cC_\ell$D} is NP-hard for each $\ell\geq 3$ and each $c\in[\ell]$ even if~$G$ has girth~$\ell$ and each $C_\ell$ in $G$ is $c$-colored on strictly non-cascading graphs.
\end{theorem}

\section{Parameterized Complexity}
\label{sec:fpt}
Motivated by the NP-hardness of \textsc{$cP_\ell$D} and \textsc{$cC_\ell$D}, we  now study the parameterized complexity of these problems. We first show that both problems admit an FPT-algorithm for a new parameter which we call colored neighborhood diversity. Then, we outline the limits of parameterization by the  solution size~$k$.

\subsection{Parameterization by Colored Neighborhood Diversity}
We extend the notion of the well-known parameter neighborhood diversity to a similar parameter for edge-colored graphs. In uncolored graphs, two vertices~$u$ and~$v$ belong to the same neighborhood class if~$N[u]=N[v]$ or~$N(u)=N(v)$. This defines an equivalence relation over the vertex set of the graph. The \emph{neighborhood diversity} is then defined as the number of equivalence classes induced by this relation~\cite{L12}. 
We show that all problems~\textsc{$cP_\ell$D} with~$\ell \geq 3$ and~\textsc{$cC_\ell$D} with~$\ell \geq 5$ are fixed-parameter tractable for a parameter we call \emph{colored neighborhood diversity}.

\begin{definition} \label{Def:ColNeighDiv}
Let~$G$ be a $c$-colored graph, and let~$u,v$ be vertices of~$G$. We say that~$u$ and~$v$ belong to the same \emph{colored neighborhood class} if either
\begin{enumerate}
\item[a)] $N^i[u]=N^i[v]$ for some color~$i$, and~$N^j(u) = N^j(u)$ for every other color~$j \in [c] \setminus \{i\}$, or
\item[b)] $N^j(u) = N^j(v)$ for every color~$j \in [c]$.
\end{enumerate}
If~$u$ and~$v$ belong to the same colored neighborhood class we write~$u \sim v$.
\end{definition}

We define the~\emph{colored neighborhood diversity}~$\gamma:= \gamma(G)$ as the number of equivalence classes induced by~$\sim$. 
To see that this definition is sound we first show that~$\sim$ is in fact an equivalence relation.
\begin{proposition}

Let~$G$ be a~$c$-colored graph. Then, $\sim$ is an equivalence relation on the vertices of~$G$.
\end{proposition}
\begin{proof}
By definition,~$\sim$ is reflexive and symmetric. It remains to show transitivity. Let~$u \sim v$ and~$v \sim w$. Consider the following cases.

\textbf{Case 1: $\{u,v\}\notin E$ and $\{v,w\} \not\in E$.} Then, $u \sim v$ and~$v \sim w$ implies~$N^i(u)=N^i(v)=N^i(w)$ for each color~$i$. Thus, we have~$u \sim w$.

\textbf{Case 2: $\{u,v\} \in E_i$ for some color~$i$.} Then,~$v \sim w$ implies~$\{u,w\} \in E_i$ and~$u \sim v$ implies~$\{v,w\} \in E_i$. Consequently, it holds that~$N^i[u]=N^i[v]=N^i[w]$ and~$N^j(u)=N^j(v)=N^j(w)$ for all~$j \in[c] \setminus \{i\}$. Thus, we have~$u \sim w$. 
\end{proof}

We refer to the equivalence classes of~$\sim$ as \emph{colored neighborhood classes}. Observe that each colored neighborhood class~$K$ is either an independent set or a clique where all edges of~$E(K)$ have the same color. Moreover, observe that the neighborhood~$N(K)$ can be partitioned into non-empty vertex sets~$K_1', K_2' , \dots, K_t'$ such that each~$K_i'$ forms a colored neighborhood class in~$G$. Given~$K$, we let~$\mathcal{N}(K):=\{K_1', \dots, K_t'\}$ denote the set of these colored neighborhood classes.

Throughout this section, we call a graph~$F$~\emph{color diverse} \todog{Name "color diverse"?} if it holds that~$|K|=1$ for every colored neighborhood class~$K$ of~$F$. Let~$\mathcal{F}$ be a set of~$c$-colored graphs such that every~$F \in \mathcal{F}$ is color diverse and it can be checked in polynomial time whether a graph contains some~$F \in \mathcal{F}$ as induced subgraph. We show that, in this case,~\textsc{$\mathcal{F}$-Deletion} is fixed-parameter tractable when parameterized by~$\gamma$. For~$c \geq 2$, this implies fixed-parameter tractability of~\textsc{$cP_\ell$D} with~$\ell \geq 3$ and~\textsc{$cC_\ell$D} with~$\ell \geq 5$ since these problems can be modeled as a special case of~$\mathcal{F}$-Deletion since~$c$-colored~$P_\ell$s with~$\ell \geq 3$ and~$c$-colored~$C_\ell$s with~$\ell \geq 5$ are color diverse. The following definition is important for our fixed-parameter algorithm.

\begin{definition}
Let~$G=(V,E)$ be a~$c$-colored graph, let~$K$ be a colored neighborhood class, and let~$S \subseteq E$ be an edge-deletion set. Then,~$S$ is called~\emph{consistent with~$K$} if for every vertex~$v \in N(K)$ either~$E(\{v\},K) \subseteq S$ or~$E(\{v\},K) \cap S = \emptyset$.
\end{definition}

Intuitively, a vertex deletion set~$S$ is consistent with a colored neighborhood class all its vertices behave in the same way with respect to~$S$. 
The fixed-parameter algorithm exploits that there is always a solution~$S$ that is consistent with every colored neighborhood class. 

\begin{lemma} \label{Proposition: Equal Choice Neighborhood Class}
Let~$\mathcal{F}$ be a set of~$c$-colored graphs such that each colored neighborhood class of every~$F \in \mathcal{F}$ has size one. Let~$G$ be a~colored graph and let~$K$ be a colored neighborhood class of~$G$. Moreover, let~$S$ be an edge-deletion set such that~$G-S$ has no induced~$F\in \mathcal{F}$. Then, there exists an edge-deletion set~$S'$ with~$|S'|\le |S|$ such that
\begin{enumerate}
\item[a)] $S'$ is consistent with~$K$ 
\item[b)] $G-S'$ is~$F$-free for every~$F \in \mathcal{F}$,
\item[c)] for each~$e \not \in E(K) \cup E(K, N(K))$, we have~$e\in S'$ if and only if~$e \in S$.
\item[d)] if~$S$ is consistent with a class~$K' \in \mathcal{N}(K)$, then~$S'$ is consistent with~$K'$.
\end{enumerate}
\end{lemma}
\begin{proof}
We transform~$S$ into~$S'$: For each~$u \in K$ we define~$S_u:=\{w \in N(K) \mid \{u,w\} \in S\}$. Let~$v \in K$ such that~$|S_v| = \min_{u \in K} |S_u|$. The new edge-deletion set~$S'$ contains the same edges as~$S$ but for every~$u \in K \setminus \{v\}$ we do the following: for every~$w \in S_v$ we add the edge~$\{u,w\}$ and for every~$w \in N(K) \setminus S_v$ we remove the edge~$\{u,w\}$. Moreover, if~$K$ is a clique, we remove all edges in~$E_G(K)$.

Since~$|S_v|$ is minimal, it follows that~$|S'| \leq |S|$. Moreover, Statement~$c)$ holds by the construction of~$S'$. For Statement~$a)$, let~$w \in N(K)$. If~$\{v,w\} \in S$ it follows that~$E(\{w\},K) \subseteq S'$. Otherwise, if~$\{v,w\} \not \in S$ it follows that~$E(\{w\},K) \cap S' = \emptyset$. Thus, Statement~$a)$~holds.
We next show Statement~$b)$. That is, we show that~$G-S'$ is~$F$-free for every~$F \in \mathcal{F}$. To this end, observe that~$K$ is a colored neighborhood class in~$G-S'$ by the construction of~$S'$. 
Let~$e := \{u,w\} \in E(K) \cup E(K, N(K))$. Assume towards a contradiction that there exists a vertex set~$Z$ containing~$u$ and~$w$ such that~$(G-S')[Z]$ is an induced~$F \in \mathcal{F}$. We show that the following two cases are contradictory.

\textbf{Case 1: $|Z \cap K| \geq 2$.} Then, since~$K$ is a colored neighborhood class in~$G-S'$, this contradicts the fact that~$F$ has only colored neighborhood classes of size one. 

\textbf{Case 2:~$|Z \cap K| = 1$.} Without loss of generality let~$Z \cap K= \{u\}$. We then define~$Z':= Z \setminus \{u\} \cup \{v\}$. Observe that in~$G-S$, vertex~$v$ has the same colored neighbors in~$N(K)$ as~$u$ in~$G-S'$. We conclude that~$G-S[Z \setminus \{u\} \cup \{v\}]$ is an induced~$F$ which contradicts the fact that~$G-S$ is~$\mathcal{F}$-free.

It remains to prove~$d)$. Let~$K' \in  \mathcal{N}(K)$ such that~$S$ is consistent with~$K'$. Statement~$c)$ implies that~$E(\{w\}, K') \subseteq S'$ or $E(\{w\}, K') \cap S'= \emptyset$ for every~$w \in N(K')\setminus K$. So, let~$u \in K$. Since~$S$ is consistent with~$K'$ it holds that~$E(\{v\}, K') \subseteq S$ or $E(\{v\}, K') \cap S= \emptyset$. Then, since in~$G-S'$ every~$u \in K$ has the same colored neighbors in~$N(K)$ as~$v$ has in~$G-S$, it follows that~$S'$ is consistent with~$K'$.
\end{proof}

In the proof of Lemma~\ref{Proposition: Equal Choice Neighborhood Class} we exploit that the neighborhood classes in every~$F \in \mathcal{F}$ have size at most one. In fact, this condition
is necessary since the lemma does not hold for~\textsc{$2C_4$D}~as we can see in the example in Fig.~\ref{Figure: C4Del}.
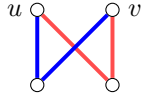
\begin{figure}[t]
\begin{center}
\begin{tikzpicture}
\tikzstyle{knoten}=[circle,fill=white,draw=black,minimum size=5pt,inner sep=0pt]
\tikzstyle{bez}=[inner sep=0pt]

\node[knoten] (u)  at (0,0) {};
\node[knoten] (v)  at (1,0) {};
\node[bez]   at (-0.3,0) {$u$};
\node[bez]   at (1.3,0) {$v$};
\node[knoten] (w)  at (0,-1) {};
\node[knoten] (x)  at (1,-1) {};

\draw[blue edge]  (u) to (w);
\draw[red edge]  (u) to (x);

\draw[blue edge]  (v) to (w);
\draw[red edge]  (v) to (x);

\end{tikzpicture}
\end{center} 
\caption{An instance of~\textsc{$2C_4$D} that has a solution of size~$1$. This single edge deletion is either incident with~$u$ or incident with~$v$.}
\label{Figure: C4Del} 
\end{figure}

 By Lemma~\ref{Proposition: Equal Choice Neighborhood Class}, we may assume that, for such~$\mathcal{F}$ a solution~$S$ of an instance of \textsc{$\mathcal{F}$-Deletion} is consistent with every colored neighborhood class. Then, for every pair~$K_1$, $K_2$ of colored neighborhood classes either~$E(K_1,K_2) \subseteq S$ or~$E(K_1,K_2) \cap S = \emptyset$. This assumption can be used to obtain the following.

\begin{theorem}
Let~$\mathcal{F}$ be a set of~$c$-colored graphs such that every~$F \in \mathcal{F}$ is~color diverse and we can check in polynomial time whether a graph contains some~$F\in \mathcal{F}$ as induced subgraph. Then, \textsc{$\mathcal{F}$-Deletion} can be solved in~$2^{\gamma^2} \cdot n^{\Oh(1)}$ time, where~$\gamma$ denotes the colored neighborhood diversity.
\end{theorem}
\begin{proof}
First, compute all colored neighborhood classes in polynomial time. Note that there are~$\Oh(\gamma^2)$ edge sets between different colored neighborhood classes. Then, iterate over all~$\Oh(2^{\gamma^2})$ possibilities to delete some of these edge sets. If one of these edge deletions leads to a solution, return \textit{yes}. Otherwise, return \textit{no}. 
\end{proof}

\begin{corollary}
\textsc{$cP_\ell$D} with~$\ell \geq 3$ and~\textsc{$cC_\ell$D} with~$\ell \geq 5$ can be solved in~$2^{\gamma^2} \cdot n^{\mathcal{O}(1)}$~time.
\end{corollary}

\subsection{Parameterization by Solution Size}

Observe that if~$\ell$ is a constant, \textsc{$cP_\ell$D} and \textsc{$cC_\ell$D} can be solved by a naive branching algorithm  with running time~$\Oh(\ell^k \cdot n^\ell)$: For a given instance~$(G,k)$ check in~$\Oh(n^\ell)$ time if~$G$ contains a~$c$-colored~$P_\ell$ (or~$C_\ell$, respectively). If this is not the case, then answer yes. Otherwise, answer no if~$k<1$. If~$k \geq 1$, then compute a~$c$-colored~$P_\ell$ (or~$C_\ell$) with edges~$e_1, \dots, e_\ell$ and branch into the cases~$(G-e_i,k-1)$ for~$i \in \{1, \dots, \ell\}$.

To put the naive branching algorithm into context, we next study both problems parameterized by~$k$ and~$\ell$ when~$\ell$ is not a constant. More precisely, we study the problems~\textsc{Colored Path Deletion (CPD)} and~\textsc{Colored Cycle Deletion (CCD)} which are versions of~$cP_\ell$D and~$cC_\ell$D where~$\ell$ and~$c$ are part of the input. For the parameter~$\ell$ it is W[1]-hard to decide whether a given graph has an induced~$P_\ell$ and to decide whether a given graph has an induced~$C_\ell$~\cite{CF07}. Consequently,~CPD and CCD are~W[1]-hard for~$\ell$ even if~$c=1$ and~$k=0$. Thus, it is hopeless to obtain fixed-parameter tractability for~$k$ or even~$k+\ell$.

The above might give the impression that the hardness of~CPD and CCD is rooted in the problem of detecting the forbidden subgraphs. However, we show that even if the forbidden subgraphs are given,~CPD and CCD are still unlikely to be fixed-parameter tractable for~$k$. More precisely, we show that both problems are~W[2]-hard when parameterized by~$k$ even if the induced subgraphs can be enumerated within polynomial time.

\begin{theorem} \label{Theorem: W2-h}
\textsc{CPD} is $W[2]$-hard when parameterized by~$k$ even if

\begin{enumerate}
\item[a)] all induced~$c$-colored~$P_\ell$s can be enumerated in polynomial time, and
\item[b)] the input is limited to instances where~$c=3$ and the input graph is non-cascading and has~$n^{\Oh(1)}$ induced~$c$-colored~$P_\ell$s.
\end{enumerate}

\end{theorem}

\begin{proof}
We give  a parameterized reduction from the~$W[2]$-hard problem \textsc{Hitting Set (HS)} parameterized by $k$~\cite{Cyg+15}. In \textsc{HS} one is given a universe~$U=\{x_1,\dots,x_\eta\}$, a family~$\mathcal{F}=\{F_1\ldots,F_\mu\}$ of subsets of~$U$, and an integer~$k$. The question is if there is some~$H \subseteq U$ with~$|H| \leq k$ and~$H \cap S \neq \emptyset$ for every~$F \in \mathcal{F}$.
Let $(U,\mathcal{F}, k)$ be an instance of \textsc{HS}.
We can assume that each set $F_j$ is non-empty and that each~$x_i\in U$ occurs in at least one subset $F_j\in\mathcal{F}$.

\textit{Construction}:
To construct an equivalent instance $(G,c,\ell,k)$ of \textsc{CPD} we first set $c=3$ and~$\ell = 1 + 3\cdot \eta$. 
Then, we construct the following gadgets.

For each $x_i\in U$ we construct an \textit{element gadget} $W_i$ as follows:
We add two vertices~$w_i, \widetilde{w}_i$ to~$W_i$ and connect $w_i$ with $\widetilde{w}_i$ by a blue edge.
By $W$ we denote the set of all element gadgets.

Next, we construct a \textit{subset gadget} $Z_j$ for each set ${F}_j\in\mathcal{F}$.
We add a vertex~$v^j$, and for each $i\in[\eta]$ a vertex $u_i^j$ to $Z_j$. 
Then, we add a yellow edge $\{v^j,u_1^j\}$.
If~$x_i\in {F}_j$, then we connect the corresponding element gadget by adding a red edge $\{u_i^j,w_i\}$ and if~$i<\eta$ a red edge $\{\widetilde{w}_i, u_{i+1}^j\}$.
Else if $x_i\notin {F}_j$, then we add two vertices $w_i^j, \widetilde{w}_i^j$ to $Z_j$, and we add red edges $\{u_i^j,w_i^j\}, \{w_i^j, \widetilde{w}_i^j\}$ to $G$ and if $i<\eta$ we add a red edge~${\{\widetilde{w}_i^j, u_{i+1}^j\}}$ to $G$.
All edges that we have added so far are called \textit{unfixed} edges.

Finally, we connect the subset gadgets as follows.
Let ${F}_p,{F}_q$ be subsets such that~$x_i\in {F}_p$ and~$x_i\in {F}_q$ for an element $x_i\in U$.
We add red edges 
$\{u_i^p,u_i^q\}$,
$\{v^p,u_i^q\}$, 
$\{v^q,u_i^p\}$, 
$\{u_1^p,u_i^q\}$,
$\{u_1^q,u_i^p\}$ 
and if $i<\eta$ we add red edges 
$\{u_{i+1}^p,u_{i+1}^q\}$, 
$\{v^p,u_{i+1}^q\}$, 
$\{v^q,u_{i+1}^p\}$, 
$\{u_1^p,u_{i+1}^q\}$, 
$\{u_1^q,u_{i+1}^p\}$ (see Fig.~\ref{fig:w2}).
We call these edges \textit{fixed}.
Observe that for any edge $e\in E$ we call $e$ a fixed edge if $e$ connects two vertices from different subset gadgets and otherwise, we call $e$ an unfixed edge.

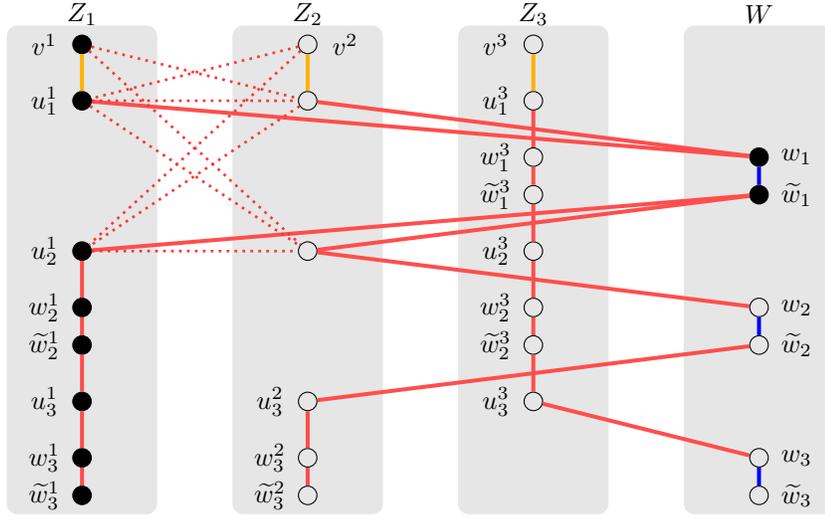
\begin{figure}[t!]
\centering
\begin{tikzpicture}[yscale=.5]
  \begin{pgfonlayer}{background}
    \begin{scope}[blend mode=multiply]
    \draw[rounded corners, fill=black!10,draw=none](-1,-.5)--(1,-.5)--(1,12.5)--(-1,12.5)--cycle;\node at(0,12.8){$Z_1$};
    \draw[rounded corners, fill=black!10,draw=none](2,-.5)--(4,-.5)--(4,12.5)--(2,12.5)--cycle;\node at(3,12.8){$Z_2$};
    \draw[rounded corners, fill=black!10,draw=none](5,-.5)--(7,-.5)--(7,12.5)--(5,12.5)--cycle;\node at(6,12.8){$Z_3$};
    \draw[rounded corners, fill=black!10,draw=none](8,-.5)--(10,-.5)--(10,12.5)--(8,12.5)--cycle;\node at(9,12.8){$W$};
      
    \end{scope}
    \end{pgfonlayer}


\node[black vertex](a1) at (9,9){};\node at(9.5,9){$w_1$};
\node[black vertex](a2) at (9,8){};\node at(9.5,8){$\widetilde{w}_1$};
\node[vertex](b1) at (9,5){};\node at(9.5,5){$w_2$};
\node[vertex](b2) at (9,4){};\node at(9.5,4){$\widetilde{w}_2$};
\node[vertex](c1) at (9,1){};\node at(9.5,1){$w_3$};
\node[vertex](c2) at (9,0){};\node at(9.5,0){$\widetilde{w}_3$};
\path[blue edge](a1)--(a2);
\path[blue edge](b1)--(b2);
\path[blue edge](c1)--(c2);

\node[black vertex](v2) at (0,10.5){};\node at(-.5,10.5){$u_1^1$};
\node[black vertex](v5) at (0,6.5){};\node at(-.5,6.5){$u_2^1$};
\node[black vertex](v6) at (0,5){};\node at(-.5,5){$w_2^1$};
\node[black vertex](v7) at (0,4){};\node at(-.5,4){$\widetilde{w}_2^1$};
\node[black vertex](v8) at (0,2.5){};\node at(-.5,2.5){$u_3^1$};
\node[black vertex](v9) at (0,1){};\node at(-.5,1){$w_3^1$};
\node[black vertex](v10) at (0,0){};\node at(-.5,0){$\widetilde{w}_3^1$};
\path[red edge](v6)--(v7);
\path[red edge](v9)--(v10);
\path[red edge](v2)--(a1);
\path[red edge](a2)--(v5);
\path[red edge](v5)--(v6);
\path[red edge](v7)--(v8);
\path[red edge](v8)--(v9);

\node[vertex](w2) at (3,10.5){};
\node[vertex](w5) at (3,6.5){};
\node[vertex](w8) at (3,2.5){};\node at(2.5,2.5){$u_3^2$};
\node[vertex](w9) at (3,1){};\node at(2.5,1){$w_3^2$};
\node[vertex](w10) at (3,0){};\node at(2.5,0){$\widetilde{w}_3^2$};
\path[red edge](w9)--(w10);
\path[red edge](w2)--(a1);
\path[red edge](a2)--(w5);
\path[red edge](w5)--(b1);
\path[red edge](b2)--(w8);
\path[red edge](w8)--(w9);

\node[vertex](x2) at (6,10.5){};\node at(5.5,10.5){$u_1^3$};
\node[vertex](x3) at (6,9){};\node at(5.5,9){$w_1^3$};
\node[vertex](x4) at (6,8){};\node at(5.5,8){$\widetilde{w}_1^3$};
\node[vertex](x5) at (6,6.5){};\node at(5.5,6.5){$u_2^3$};
\node[vertex](x6) at (6,5){};\node at(5.5,5){$w_2^3$};
\node[vertex](x7) at (6,4){};\node at(5.5,4){$\widetilde{w}_2^3$};
\node[vertex](x8) at (6,2.5){};\node at(5.5,2.5){$u_3^3$};
\path[red edge](x3)--(x4);
\path[red edge](x6)--(x7);
\path[red edge](x2)--(x3);
\path[red edge](x4)--(x5);
\path[red edge](x5)--(x6);
\path[red edge](x7)--(x8);
\path[red edge](x8)--(c1);

\node[black vertex](a1)at (0,12){};\node at(-.5,12){$v^1$};
\node[vertex](a2)at (3,12){};\node at(3.5,12){$v^2$};
\node[vertex](a3)at (6,12){};\node at(5.5,12){$v^3$};

\path[yellow edge](a1)--(v2);
\path[yellow edge](a2)--(w2);
\path[yellow edge](a3)--(x2);

\path[red dots](v2)--(w2);
\path[red dots](v5)--(w5);
\path[red dots](a1)--(w2);
\path[red dots](a1)--(w5);
\path[red dots](v2)--(w5);
\path[red dots](a2)--(v2);
\path[red dots](a2)--(v5);
\path[red dots](w2)--(v5);

\end{tikzpicture}
\caption{The constructed graph for ${\mathcal{F}=\{{F}_1=\{1\},{F}_2=\{1,2\},{F}_3=\{3\}\}}$ and~${U=[3]}$. The dotted lines represent fixed edges, while the solid lines represent unfixed edges. The filled vertices induce the 3-colored $P_{10}$ in $G[{F}_1\cup W]$.}
\label{fig:w2}
\end{figure}

\textit{Non-Cascading}: We first argue that the constructed graph is non-cascading. By construction, every unfixed edge is a conflict edge and every fixed edge is conflict-free. Moreover, every non-induced~$c$-colored~$P_\ell$ contains one yellow edge~$\{v^i,v^i_1\}$ from a subset gadget of~$Z_i$ and some vertex~$x$ from a subset gadget of some~$Z_j$ with~$j \neq i$. Then, there are fixed edges~$\{v^i,x\}$ and~$\{v^i_1,x\}$. Since all fixed edges are non-conflict edges we conclude that~$G$ is non-cascading.

\textit{Enumerating~$c$-colored~$P_\ell$s}: By construction, for every induced~$c$-colored~$P_\ell$ there is one unique~$v \in \{v^i \mid i \in [\mu]\}$ and one unique~$w \in \{\tilde{w}_\eta\} \cup \{\tilde{w}_\eta^i \mid i \in [\mu]\}$ such that~$v$ and~$w$ are the endpoints of the induced~$c$-colored~$P_\ell$. Thus, there are exactly~$\mu$ induced~$c$-colored~$P_\ell$ in~$G$ than can be enumerated in polynomial~time.

\textit{Intuition}:
Before we prove the correctness of the reduction, we describe its idea.
We connected the subset gadgets to element gadgets such that for each subset gadget~$Z_j$, the induced subgraph $G[Z_j\cup W]$ contains exactly one induced 3-colored $P_\ell$. 
We then connected the subset gadgets such that there is no induced 3-colored $P_\ell$ in $G$ that contains vertices from two different subset gadgets. 
So we can model a hitting set for a collection $\mathcal{F}$ by deleting the edges from the corresponding element gadgets. 

\textit{Correctness}:
To prove the correctness of the reduction, we show the following claims.
First, we show that each 3-colored $P_\ell$ contains vertices of exactly one subset gadget if we do not delete fixed edges.

\begin{Claim}
Let $S$ be an edge-deletion set that does not contain fixed edges.
Furthermore, let~$G':=G- S$ and $V'\subseteq V(G)$ such that the induced subgraph $G'[V']$ is a~3-colored $P_\ell$.
Then, $V'$ contains vertices of at most one subset gadget $Z_j$.
\label{Claim:w2_1}
 \end{Claim}

\begin{claimproof}
Since each 3-colored $P_\ell$ has to include a yellow edge, we conclude \linebreak[4]that~${v^p,u_1^p\in V'}$ for some $p\in[\mu]$.
Assume towards a contradiction that $V'$ contains a vertex from a clause gadget $Z_q$ such that $q\neq p$.
Without loss of generality there are adjacent vertices $\alpha, \beta\in V'$ such that $\alpha\in Z_p$, $\beta\in Z_q$ and $\{\alpha,\beta\}\in E(G)$, since we added the fixed edges.
By construction, we know that $\{v^p,\beta\},\{u_1^p,\beta\}\in E\setminus S$ since $\{v^p,\beta\},\{u_1^p,\beta\}$ are fixed edges.
This is a contradiction since the induced subgraph $G'[\{v^p, u_1^p,\beta\}]$ is a $C_3$. $\hfill \Diamond$
\end{claimproof}

Second, we observe  the following about 3-colored $P_\ell$s that are induced by one subset gadget.

\begin{Claim}
For each subset gadget $Z_j$, the graph $G[Z_j\cup W]$ contains exactly one 3-colored $P_\ell$.
This 3-colored $P_\ell$ contains a blue edge $\{w_i,\widetilde{w}_i\}$ if and only if $x_i\in {F}_j$.
\label{Claim:w2_2}
 \end{Claim}

\begin{claimproof}
By construction, we know that $G[Z_j\cup W]$ contains at least one 3-colored $P_\ell$, and that $\{v^j,u_1^j\}$ is the only yellow edge in $E_G(Z_j\cup W)$.
Since $\deg_{G[Z_j\cup W]}(v^j)=1$ and $\deg_{G[Z_j\cup W]}(\alpha)\leq 2$ for each $\alpha\in Z_j\cup W$, we conclude that there is at most one~3-colored $P_\ell$ in $G[Z_j\cup W]$.

By construction, this 3-colored $P_\ell$ contains a blue edge $\{w_i,\widetilde{w}_i\}$ if $x_i\in {F}_j$.
Otherwise, if~${x_i\notin {F}_j}$, then the 3-colored $P_\ell$ contains a red edge $\{w_i^j,\widetilde{w}_i^j\}$ and does not contain the blue edge $\{w_i,\widetilde{w}_i\}$. $\hfill \Diamond$
\end{claimproof}

Now we prove the correctness of the reduction by showing that $(U, \mathcal{F}, k)$ is a yes-instance of \textsc{HS} if and only if $(G,c,\ell,k)$ is a yes-instance of \textsc{CPD}.

($\Rightarrow$) 
Let $H\subseteq U$ be a hitting set of size at most $k$ for $(U,\mathcal{F})$.
Consider the set~$S:=\{\{w_i,\widetilde{w}_i\}\ |\ x_i\in H\}$.
We will show that~$G':= G- S$ is 3-colored $P_\ell$-free.
Since $S$ does not contain fixed edges, we know from Claim \ref{Claim:w2_1} that there is no induced 3-colored $P_\ell$ in $G'$ that contains vertices from two different subset gadgets.
Hence, it remains to show that no subgraph~$G'[Z_j\cup W]$ contains an induced 3-colored $P_\ell$ for any subset gadget $Z_j$.

Let $Z_j$ be a subset gadget.
By Claim \ref{Claim:w2_2} we know that there is exactly one 3-colored $P_\ell$ in the induced subgraph $G[Z_j\cup X]$ and that this 3-colored $P_\ell$ includes exactly one blue edge from each element gadget~$W_i$ where $x_i\in {F}_j$.
Since $H$ is a hitting set for $\mathcal{F}$, we conclude that $S$ includes at least one blue edge from an element gadget~$W_i$ that is part of the 3-colored $P_\ell$ in $G[Z_j\cup W]$.
Hence,~$G'[Z_j\cup W]$ is 3-colored~$P_\ell$-free, which implies that $G'$ is 3-colored $P_\ell$-free.

($\Leftarrow$)
Conversely, let $S$ be an edge-deletion set of size at most $k$ such that $G- S$ is 3-colored $P_\ell$-free. 
First, we show that $S':=S\setminus \{e\in S\ |\ e\text{ is fixed}\}$ is an equivalent solution. 
Assume towards a contradiction that $G':=G- S'$ contains an induced 3-colored $P_\ell$. 
Let $V'\subseteq V(G)$ such that $G'[V']$ is a 3-colored~$P_\ell$.
Since~$S'$~does not contain fixed edges, we know by Claim \ref{Claim:w2_1} that~$V'$ contains vertices of at most one subset gadget $Z_j$.
Hence, $G'[V']$ only contains unfixed edges.
This implies that~$G'[V']$ does not contain an edge from $S$.
Hence, the induced subgraph~$G[V']$ is a 3-colored $P_\ell$ in $G'$.
This is a contradiction, since $G- S$ is 3-colored $P_\ell$-free.
Hence,~$S'$ is an equivalent solution.

Next, we will show how to construct another equivalent solution $S''\subseteq E_G(W)$ from $S'$.
Let~$\{\alpha,\beta\}\in S'$ such that $\{\alpha,\beta\}\notin E_G(W)$.
We will show that $\{\alpha,\beta\}$ is part of at most one 3-colored $P_\ell$.
From Claim~\ref{Claim:w2_1}, we conclude that any 3-colored~$P_\ell$ including $\{\alpha,\beta\}$ only includes vertices from at most one subset gadget, since $S'$ only contains unfixed edges.
Hence, $\{\alpha,\beta\}$ can only be part of a 3-colored $P_\ell$ in an induced subgraph $G[Z_j\cup W]$ where $\alpha\in Z_j$ or $\beta\in Z_j$.
From Claim \ref{Claim:w2_2} we know that there is exactly one 3-colored $P_\ell$ in $G[Z_j\cup W]$.
Hence, $\{\alpha,\beta\}$ is part of at most one 3-colored $P_\ell$.
This 3-colored $P_\ell$ has to include a blue edge $e\in E_G(W)$, since $G$ is 3-colored and the edges in $E_G(W)$ are the only blue edges.
Hence,~$(S\setminus \{\{\alpha,\beta\}\})\cup\{e\}$ is an equivalent solution.
Thus, we can construct an equivalent solution $S''$ that only contains edges from $E_G(W)$. 

To finish the proof, we show that the set $H:=\{x_i\ |\ \{w_i,\widetilde{w}_i\}\in S''\}$ is a hitting set for $\mathcal{F}$.
Let ${F}_j\in\mathcal{F}$. 
From Claim~\ref{Claim:w2_2} we know that $G[Z_j\cup W]$ includes a 3-colored $P_\ell$. 
Hence, there is at least one edge from that 3-colored $P_\ell$ in every solution.
Since $S''$ is a solution that only contains edges~$\{w_i,\widetilde{w}_i\}\in E_G(W)$, we conclude that there is at least one $x_i\in H$ such that $x_i\in {F}_j$.
Hence, $H$ is a hitting set for~$\mathcal{F}$. 
\end{proof} 

We next extend the construction from the proof of Theorem~\ref{Theorem: W2-h} to obtain a similar W[2]-hardness for~\textsc{CCD}.
\begin{theorem} \textsc{CCD} is $W[2]$-hard when parameterized by~$k$ even if
  \begin{enumerate}
  \item[a)] all induced~$c$-colored~$C_\ell$s can be enumerated in polynomial~time, and
  \item[b)] the input is limited to instances where~$c=4$ and the input graph is non-cascading and has~$n^{\Oh(1)}$ induced~$c$-colored~$C_\ell$s.
  \end{enumerate}
\end{theorem}

\begin{proof}
Let~$I:=(G,c=3,\ell,k)$ be an instance that has been constructed from an instance of HS according to the proof of Theorem~\ref{Theorem: W2-h}. Recall that~$G$ is non-cascading and that every induced~$c$-colored~$P_\ell$ in~$G$ there is one unique~$v \in A:= \{ v^i \mid i \in [\mu]\}$ and one unique~$w \in B:=\{\tilde{w}_\eta\} \cup \{\tilde{w}_\eta^i \mid i \in [\mu]\}$ that are the endpoints of the induced~$c$-colored~$P_\ell$.

We describe how to construct an instance~$I':=(G',c=4,\ell+1,k)$ such that~$I'$ is a yes-instance of~CCD if and only if~$I$ is a yes-instance of~CPD.

\textit{Construction:} Let~$1$, $2$, and~$3$ be the colors present in~$G$. We describe how to obtain~$G'$ that is colored with edge colors~$1$, $2$, $3$, and~$4$. 
The graph~$G'$ can be computed from~$G$ by adding a vertex set~$Q$ of size~$k+1$, and we add edges with color~$4$ between~$Q$ and~$A \cup B$ such that~$Q \cup (A \cup B)$ forms a biclique. 

Analogously to Theorem~\ref{Theorem: W2-h}, every non-induced~$4$-colored~$P_{\ell+1}$ contains vertices from different subset gadgets of~$G$ and thus,~$G'$ is non-cascading. Moreover, observe that every induced~$4$-colored~$C_{\ell+1}$ in~$G'$ consists of a~$3$-colored~$P_\ell$ with one endpoint~$v \in A$ and one endpoint~$w \in B$ and two edges~$\{v,q\}$, $\{q,w\}$ with~$q \in Q$. Thus, there are~$\mu \cdot (k+1)$ induced~$4$-colored~$C_{\ell+1}$ in~$G'$ which can be enumerated by enumerating all all~$3$-colored~$P_\ell$ in~$G$ in~$\text{poly}(n)$ time and combining each with one vertex from~$Q$.

\textit{Correctness:} We next show the correctness.

$(\Rightarrow)$ Let~$S$ with~$|S| \leq k$ be an edge-deletion set such that~$G-S$ is~$3$-colored~$P_\ell$-free. 
Then, in~$G'-S$ there is no induced~$3$-colored~$P_\ell$ with one endpoint in~$A$ and one endpoint in~$B$. 
Therefore, there is no induced~$4$-colored~$C_{\ell+1}$ in~$G'-S$. 
Thus,~$S$ is a solution of size at most~$k$ for~$I'$.

$(\Leftarrow)$ Let~$S'$ with~$|S'| \leq k$ be an edge-deletion set such that~$G'-S$ is~$4$-colored~$C_{\ell+1}$-free. 

First, assume towards a contradiction that there is some~$v \in A$ and some~$w \in B$ such that there is an induced~$3$-colored~$P_\ell$ that starts in~$v$, ends in~$w$ and only contains vertices from~$G$. 
Note that there are~$k+1$ edge-disjoint paths~$(v,q,w)$ with~$q \in Q$. 
Each such path has subsequent colors~$4,4$ and forms an induced~$4$-colored~$C_{\ell+2}$ with the~$c$-colored~$P_\ell$ connecting~$v$ and~$w$. 
Consequently, for each such path~$(v,q,w)$ at least one edge on the path belongs to~$S'$ contradicting the fact that~$|S'| \leq k$.

We next show that this implies that~$I$ has a solution of size at most~$k$. Let~$X$ be the set of conflict edges in~$G$, which is the set of all edges of~$G$ that are part of an induced~$c$-colored~$P_\ell$. Furthermore, let~$S \subseteq S'$ be the set of edge-deletions between the vertices of~$G$. By the above, every induced~$c$-colored~$P_\ell$ in~$G$ is not an induced~$c$-colored~$P_\ell$ in~$G-S$. Then, since~$G$ is non-cascading, we conclude by Proposition~\ref{Prop: Non-Cascading} that~$G-(S \cap X)$ has no induced~$c$-colored~$P_\ell$. Thus, $S \cap X$ is a solution of size at most~$k$ for~$I$. 
\end{proof}


\section{$2P_4$ Deletion on Double Cluster Graphs}
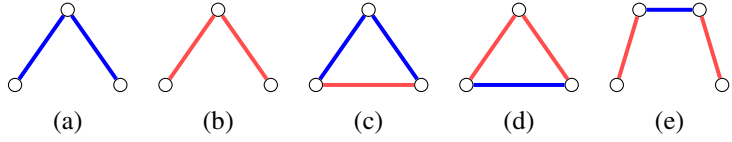
\begin{figure}[t]
\begin{center}
\begin{tikzpicture}
\tikzstyle{knoten}=[circle,fill=white,draw=black,minimum size=5pt,inner sep=0pt]
\tikzstyle{bez}=[inner sep=0pt]

\node[knoten] (u)  at (0,0) {};
\node[knoten] (v)  at (0.7,1) {};
\node[knoten] (w)  at (1.4,0) {};
\node[bez] at (0.7,-0.5) {(a)};

\draw[blue edge]  (u) to (v);
\draw[blue edge]  (w) to (v);

\begin{scope}[xshift=2cm]
\node[knoten] (u)  at (0,0) {};
\node[knoten] (v)  at (0.7,1) {};
\node[knoten] (w)  at (1.4,0) {};
\node[bez] at (0.7,-0.5) {(b)};

\draw[red edge]  (u) to (v);
\draw[red edge]  (w) to (v);
\end{scope}

\begin{scope}[xshift=4cm]
\node[knoten] (u)  at (0,0) {};
\node[knoten] (v)  at (0.7,1) {};
\node[knoten] (w)  at (1.4,0) {};
\node[bez] at (0.7,-0.5) {(c)};

\draw[blue edge]  (u) to (v);
\draw[blue edge]  (v) to (w);
\draw[red edge]  (w) to (u);
\end{scope}

\begin{scope}[xshift=6cm]
\node[knoten] (u)  at (0,0) {};
\node[knoten] (v)  at (0.7,1) {};
\node[knoten] (w)  at (1.4,0) {};
\node[bez] at (0.7,-0.5) {(d)};

\draw[red edge]  (u) to (v);
\draw[red edge]  (w) to (v);
\draw[blue edge]  (w) to (u);
\end{scope}

\begin{scope}[xshift=8cm]
\node[knoten] (u)  at (0,0) {};
\node[knoten] (v)  at (0.3,1) {};
\node[knoten] (x)  at (1.1,1) {};
\node[knoten] (w)  at (1.4,0) {};
\node[bez] at (0.7,-0.5) {(e)};

\draw[red edge]  (u) to (v);
\draw[blue edge]  (v) to (x);
\draw[red edge]  (x) to (w);
\end{scope}

\end{tikzpicture}
\caption{The double cluster graphs are the graphs that have none of~(a) to~(d) as induced subgraph; $\mathcal{T}$ is the class of graphs that do not have~(a) to~(e) as induced~subgraph.} \label{Figure: forbidden subgraphs}
\end{center} 
\end{figure}

We now study~\textsc{$2P_4$D} on bicolored graphs where both the blue and the red subgraph are a cluster graph; we refer to such graphs as \emph{double cluster graphs}. Equivalently, a graph is a double cluster graph if it has no graph from Fig.~\ref{Figure: forbidden subgraphs} (a)--(d) as induced subgraph. 
As we show,~$2P_4$D remains NP-hard on double cluster graphs. This is in sharp contrast to \textsc{$2P_3$D} which is solvable in polynomial time in this graph class~\cite{gruettemeier19}. 
This hardness is tight in the following sense: In double cluster graphs, there exist only two types of induced bicolored~$P_4$s, the~$P_4$ with subsequent edge colors red, blue, red and the~$P_4$ with subsequent edge colors blue, red, blue. If the input is a double cluster graph where one of these bicolored~$P_4$ is forbidden, say the~$P_4$ shown in Fig.~\ref{Figure: forbidden subgraphs}~(e), then~$2P_4$D can be solved in polynomial time.
In the following we denote this graph class by~$\mathcal{T}$.

\begin{theorem}
\textsc{$2P_4$D} remains NP-hard on bicolored graphs where both the blue and the red subgraph are a cluster graph and the maximum degree is five.
\end{theorem}

\begin{proof}
We give a polynomial-time reduction from the NP-complete \textsc{(3,B2)-SAT} problem~\cite{berman03}. 
Recall that \textsc{(3,B2)-SAT} is a version of \textsc{3SAT} where one is given a CNF formula~$\Phi$ on variables~$x_1, \dots, x_\eta$ where every clause~$c_1,\dots,c_{\mu}$ contains exactly three literals and each literal~$x_i$ and~$\neg x_i$ occurs exactly twice in~$\Phi$.

\textit{Construction:} The instance~$(G,k)$ of \textsc{$2P_4$D} consists of one clause gadget~$Z_i$ for each clause~$c_i$ and one variable gadget~$X_j$ for each variable~$x_j$. 
The clause gadget~$Z_i$ contains a blue triangle with vertices~$d_i^1, d_i^2$, and~$d_i^3$.
Furthermore, for each~$z\in[3]$ we attach one leaf vertex~$c_i^z$ to vertex~$d_i^z$ in the triangle. 
Each edge~$\{c_i^z,d_i^z\}$ is colored red.

The variable gadget~$X_j$ is constructed as follows:

\begin{itemize}
\item We add a red clique on four vertices~$r_j^1$, $r_j^2$, $r_j^3$, and~$r_j^4$ to~$G$.
\item We add six blue triangles~$\{r^1_j,p^1_j,p^2_j\}$,~$\{r^2_j,q^1_j,q^2_j\}$,~$\{p^3_j,p^5_j,p^7_j\}$,~$\{q^3_j,q^5_j,q^7_j\}$, $\{p^4_j,p^6_j,p^8_j\}$, and $\{q^4_j,q^6_j,q^8_j\}$ to~$G$.
\item We add four blue edges~$\{p^9_j,t^1_j\}$,~$\{p^{10}_j,t^2_j\}$,~$\{q^9_j,f^1_j\}$, and~$\{q^{10}_j,f^2_j\}$ to~$G$.
\item We add the eight red edges~$\{p^1_j,p^3_j\}$,~$\{p^2_j,p^4_j\}$,~$\{p^5_j,p^9_j\}$,~$\{p^6_j,p^{10}_j\}$,~$\{q^1_j,q^3_j\}$, $\{q^2_j,q^4_j\}$,~$\{q^5_j,q^9_j\}$, and~$\{q^6_j,q^{10}_j\}$ to~$G$.
\end{itemize}

\begin{figure}[t]
\begin{center}
\begin{tikzpicture}
\tikzstyle{knoten}=[circle,fill=white,draw=black,minimum size=5pt,inner sep=0pt]
\tikzstyle{bez}=[inner sep=0pt]

\node[knoten] (r3)  at (0,0) {};\node at(0,-0.3){$r^3_j$};
\node[knoten] (r1)  at (0,1) {};\node at(0,1.3){$r^1_j$};
\node[knoten] (r4)  at (0.9,0) {};\node at(0.9,-0.3){$r^4_j$};
\node[knoten] (r2)  at (0.9,1) {};\node at(0.9,1.3){$r^2_j$};

\node[knoten] (p1)  at (-0.9,0) {};\node at(-0.9,-0.3){$p^2_j$};
\node[knoten] (p2)  at (-0.9,1) {};\node at(-0.9,1.3){$p^1_j$};
\node[knoten] (p3)  at (-1.8,0) {};\node at(-1.75,-0.3){$p^4_j$};
\node[knoten] (p4)  at (-1.8,1) {};\node at(-1.75,1.3){$p^3_j$};
\node[knoten] (p5)  at (-2.7,0) {};\node at(-2.9,-0.3){$p^6_j$};
\node[knoten] (p6)  at (-2.7,1) {};\node at(-2.9,1.3){$p^5_j$};
\node[knoten] (p7)  at (-2.25,-0.7) {};\node at(-2.25,-1){$p^8_j$};
\node[knoten] (p8)  at (-2.25,1.7) {};\node at(-2.25,2){$p^7_j$};
\node[knoten] (p9)  at (-3.6,0) {};\node at(-3.7,-0.3){$p^{10}_j$};
\node[knoten] (p10)  at (-3.6,1) {};\node at(-3.7,1.3){$p^9_j$};
\node[knoten] (t1)  at (-4.5,0) {};\node at(-4.6,0.3){$t^2_j=c^3_i$};
\node[knoten] (t2)  at (-4.5,1) {};\node at(-4.6,1.3){$t^1_j$};

\draw[red edge]  (r1) to (r2);
\draw[red edge]  (r1) to (r3);
\draw[red edge]  (r1) to (r4);
\draw[red edge]  (r2) to (r3);
\draw[red edge]  (r2) to (r4);
\draw[red edge]  (r3) to (r4);

\draw[red edge]  (p1) to (p3);
\draw[red edge]  (p2) to (p4);
\draw[red edge]  (p5) to (p9);
\draw[red edge]  (p6) to (p10);

\draw[blue edge]  (r1) to (p1);
\draw[blue edge]  (r1) to (p2);
\draw[blue edge]  (p1) to (p2);
\draw[blue edge]  (p3) to (p5);
\draw[blue edge]  (p3) to (p7);
\draw[blue edge]  (p5) to (p7);
\draw[blue edge]  (p4) to (p6);
\draw[blue edge]  (p4) to (p8);
\draw[blue edge]  (p6) to (p8);
\draw[blue edge]  (t2) to (p10);
\draw[blue edge]  (t1) to (p9);

\node[knoten] (q1)  at (1.8,0) {};\node at(1.8,-0.3){$q^2_j$};
\node[knoten] (q2)  at (1.8,1) {};\node at(1.8,1.3){$q^1_j$};
\node[knoten] (q3)  at (2.7,0) {};\node at(2.55,-0.3){$q^4_j$};
\node[knoten] (q4)  at (2.7,1) {};\node at(2.65,1.3){$q^3_j$};
\node[knoten] (q5)  at (3.6,0) {};\node at(3.78,-0.3){$q^6_j$};
\node[knoten] (q6)  at (3.6,1) {};\node at(3.78,1.3){$q^5_j$};
\node[knoten] (q7)  at (3.15,-0.7) {};\node at(3.15,-1){$q^8_j$};
\node[knoten] (q8)  at (3.15,1.7) {};\node at(3.15,2){$q^7_j$};
\node[knoten] (q9)  at (4.5,0) {};\node at(4.5,-0.3){$q^{10}_j$};
\node[knoten] (q10)  at (4.5,1) {};\node at(4.5,1.3){$q^9_j$};
\node[knoten] (f1)  at (5.4,0) {};\node at(5.4,-0.3){$f^2_j$};
\node[knoten] (f2)  at (5.4,1) {};\node at(5.4,1.3){$f^1_j$};

\draw[red edge]  (q1) to (q3);
\draw[red edge]  (q2) to (q4);
\draw[red edge]  (q5) to (q9);
\draw[red edge]  (q6) to (q10);

\draw[blue edge]  (r2) to (q1);
\draw[blue edge]  (r2) to (q2);
\draw[blue edge]  (q1) to (q2);
\draw[blue edge]  (q3) to (q5);
\draw[blue edge]  (q3) to (q7);
\draw[blue edge]  (q5) to (q7);
\draw[blue edge]  (q4) to (q6);
\draw[blue edge]  (q4) to (q8);
\draw[blue edge]  (q6) to (q8);
\draw[blue edge]  (f2) to (q10);
\draw[blue edge]  (f1) to (q9);

\node[knoten] (u)  at (-6.3,0) {};\node at(-6.3,-0.3){$d^1_i$};
\node[knoten] (v)  at (-5.85,1) {};\node at(-5.58,1){$d^2_i$};
\node[knoten] (w)  at (-5.4,0) {};\node at(-5.4,-0.3){$d^3_i$};
\node[knoten] (a)  at (-6.3,1) {};\node at(-6.3,1.3){$c^1_i$};
\node[knoten] (b)  at (-5.85,2) {};\node at(-5.85,2.3){$c^2_i$};

\draw[blue edge]  (u) to (v);
\draw[blue edge]  (w) to (v);
\draw[blue edge]  (w) to (u);
\draw[red edge]  (u) to (a);
\draw[red edge]  (v) to (b);
\draw[red edge]  (w) to (t1);
\end{tikzpicture}
\caption{Visualization of the construction for some clause~$c_i$ and a variable~$x_j$ where~$x_j$ has its second occurrence as the third literal in~$c_i$.} \label{Figure: double cluster graphs}
\end{center} 
\end{figure}
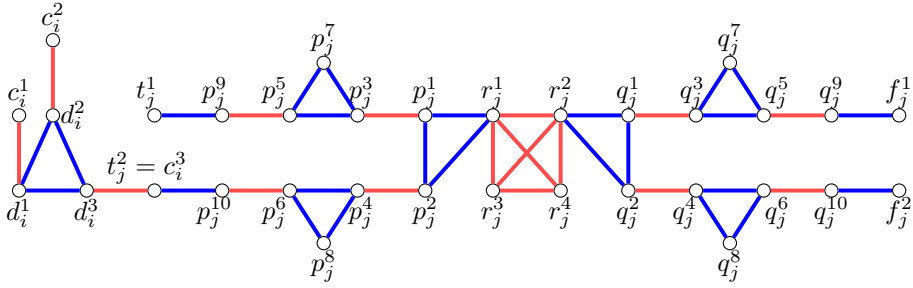

To connect the variable and the clause gadgets we identify vertices as follows. For~$z\in[3]$,~$y\in[2]$, any clause~$c_i$, and any variable~$x_j$ we set
\begin{center}
$c_i^z= 
\begin{cases}
t_j^y&\text{if the literal }x_j\text{ has its }y\text{-th occurence as the }z\text{-th literal in }c_i,\\
f_j^y&\text{if the literal }\lnot x_j\text{ has its }y\text{-th occurence as the }z\text{-th literal in }c_i.
\end{cases}$
\end{center}

Observe that since each clause contains exactly three literals and since each literal~$x_j$ and~$\lnot x_j$ occurs exactly twice in~$\Phi$, each vertex~$c_i^z$ is identified with exactly one vertex~$t_j^y$ or~$f_j^y$ and vice versa.
For an illustration of this construction see Fig.~\ref{Figure: double cluster graphs}.

It is easy to see that the maximum degree of~$G$ is five, that the blue subgraph of~$G$ is a cluster graph, and that the red subgraph of~$G$ is a cluster graph.
We complete the construction by setting~$k:=9\eta +2\mu$. 

\textit{Intuition:} 
Before we prove the correctness, we informally describe its idea. 
Since the budget is tight, exactly two edges of each clause gadget and exactly nine edges of each variable gadget will be deleted. 
There are exactly two possibilities to delete nine edges in a variable gadget to make it conflict-free.
One possibility corresponds to setting this variable to \texttt{true} and the other possibility corresponds to setting this variable to \texttt{false}. 
Furthermore, one remaining red edge in a clause gadget will represent the literal fulfilling that clause.

Before we prove the correctness of the reduction, we make the following observations.
\begin{Claim}
\label{Claim-clause-gadget-2-col-p4-del}
For any edge-deletion set~$S$ of~$G$ and any clause gadget~$Z_i$ we have~$|S\cap E(Z_i)|\ge 2$.
Furthermore, if~$|S\cap E(Z_i)|=2$, then~$|S\cap E_r(Z_i)|=2$.
\end{Claim}
\begin{claimproof}
Consider the three bicolored~$P_4$s:~$P_1:=\{c^1_i,d^1_i,d^2_i,c^2_i\}$,~$P_2:=\{c^1_i,d^1_i,d^3_i,c^3_i\}$, and~$P_3:=\{c^2_i,d^2_i,d^3_i,c^3_i\}$.
Note that no edge is contained in all three~$P_4$s.
Thus,~$|S\cap E(Z_i)|\ge 2$.

Next, consider that case that~$|S\cap E(Z_i)|= 2$.
Assume towards a contradiction that~$S$ contains a blue edge of~$Z_i$.
Without loss of generality, assume that~$\{d^1_i,d^2_i\}\in S$.
Then,~$G\setminus\{\{d^1_i,d^2_i\}\}$ contains the three bicolored~$P_4$s:~$P_2$,~$P_3$, and~$\{d^1_i,d^3_i,d^2_i,c^2_i\}$.
Note that no edge is contained in all three~$P_4$s.
Hence,~$|S\cap E(Z_i)|\ge 3$, a contradiction.
Thus,~$|S\cap E_r(Z_i)|=2$.
$\hfill \Diamond$
\end{claimproof}

\begin{Claim}
\label{Claim-variable-gadget-2-col-p4-del}
For any edge-deletion set~$S$ of~$G$ and any variable gadget~$X_j$ we have~$|S\cap E(X_j)|\ge 9$.
Furthermore, if~$|S\cap E(X_j)|=9$, then there exists an edge-deletion set~$S'$ of~$G$ such that~$|S'\cap E(X_j)|=9$ and either
\begin{enumerate}
\item $S'\cap E(X_j)=\mathcal{E}_1$ which consists of the edges of the blue triangle~$\{r_j^2,q_j^1,q_j^2\}$, the four red edges~$\{p^1_j,p^3_j\}$,~$\{p^2_j,p^4_j\}$,~$\{q^5_j,q^9_j\}$, and~$\{q^6_j,q^{10}_j\}$ and the two blue edges~$\{p_j^9,t_j^1\}$ and~$\{p_j^{10},t_j^2\}$, or
\item $S'\cap E(X_j)=\mathcal{E}_2$ which consists of the edges of the blue triangle~$\{r_j^1,p_j^1,p_j^2\}$, the four red edges~$\{q^1_j,q^3_j\}$,~$\{q^2_j,q^4_j\}$,~$\{p^5_j,p^9_j\}$, and~$\{p^6_j,p^{10}_j\}$ and the two blue edges~$\{q_j^9,f_j^1\}$ and~$\{q_j^{10},f_j^2\}$.
\end{enumerate}
\end{Claim}
\begin{claimproof}
The fact that~$|S\cap E(X_j)|\ge 9$ follows from the existence of an edge-disjoint bicolored~$P_4$ packing~$\mathcal{P}$ containing the following nine~$P_4$s. These are~$\{t^1_j,p^9_j,p^5_j,p^7_j\}$, $\{t^2_j,p^{10}_j,p^6_j,p^8_j\}$, $\{f^1_j,q^9_j,q^5_j,q^7_j\}$, $\{f^2_j,q^{10}_j,q^6_j,q^8_j\}$, $\{p^7_j,p^3_j,p^1_j,r^1_j\}$, $\{q^7_j,q^3_j,q^1_j,r^2_j\}$, $\{p^8_j,p^4_j,p^2_j,p^1_j\}$, $\{q^8_j,q^4_j,q^2_j,q^1_j\}$ and~$\{p^2_j,r^1_j,r^2_j,q^2_j\}$ in~$G$.

Next, we consider an edge-deletion set~$S$ of~$G$ with~$|S\cap E(X_j)|=9$.
First, we prove that $\{r^1_j,p^1_j\}$, $\{r^1_j,p^2_j\}\in S$ or~$\{r^2_j,q^1_j\},\{r^2_j,q^2_j\}\in S$.
Assume towards a contradiction that this is not the case.
Thus, assume without loss of generality that~$\{r^1_j,p^1_j\},\{r^2_j,q^1_j\}\in G-S$.
Hence,~$\{r^1_j,r^2_j\}\in S$.
Next, observe that~$G-\{\{r^1_j,r^2_j\}\}$ contains nine bicolored~$P_4$s. These are~$\{p^1_j,r^1_j,r^3_j,r^2_j\}$, $\{q^1_j,r^2_j,r^4_j,r^1_j\}$, $\{t^1_j,p^9_j,p^5_j,p^7_j\}$,~$\{t^2_j,p^{10}_j,p^6_j,p^8_j\}$, $\{f^1_j,q^9_j,q^5_j,q^7_j\}$, $\{f^2_j,q^{10}_j,q^6_j,q^8_j\}$, $\{p^7_j,p^{3}_j,p^1_j,r^1_j\}$, $\{p^8_j,p^{4}_j,p^2_j,r^1_j\}$, and finally $\{q^7_j,q^{3}_j,q^1_j,r^2_j\}$ only sharing the edges~$\{r^1_j,p^1_j\}$, and~$\{r^2_j,q^1_j\}$.
Thus,~$|S\cap E(X_j)|\ge 10$, a contradiction.

Hence, in the following we assume without loss of generality that~$\{r^2_j,q^1_j\}\in S$ and that~$\{r^2_j,q^2_j\}\in S$.
Next, assume towards a contradiction that~$\{q^5_j,q^9_j\}\notin S$. 
Since~$\{q^3_j,q^5_j,q^9_j,f^1_j\}$ is a bicolored~$P_4$ only sharing the edges~$\{q^5_j,q^9_j\}$ and~$\{q^9_j,f^1_j\}$ with any bicolored~$P_4$ in~$\mathcal{P}$, we conclude that~$\{f^1_j,q^9_j\}\in S$.
But then,~$\mathcal{P}\setminus\{\{q^8_j,q^4_j,q^2_j,q^1_j\},\{f^1_j,q^9_j,q^5_j,q^7_j\}\}\cup\{\{q^{10}_j,q^6_j,q^4_j,q^2_j\},\{q^2_j,q^1_j,q^3_j,q^7_j\},\{q^7_j,q^5_j,q^9_j,f^1_j\}\}$ is an edge-disjoint packing of size nine, a contradiction.
Hence,~$\{q^5_j,q^9_j\}\in S$.
Analogously, we can show that~$\{q^6_j,q^{10}_j\}\in S$.
Since~$\{q^7_j,q^3_j,q^1_j,q^2_j\}$ and $\{q^8_j,q^4_j,q^2_j,q^1_j\}$ are two bicolored~$P_4$s not in the packing~$\mathcal{P}$, we conclude that~$\{q^1_j,q^2_j\}\in S$.

Hence, in the graph~$G[t^1_j,t^2_j,r^1_j,\ldots, r^4_j,p^1_j,\ldots, p^{10}_j]$ the edge-deletion set~$S$ does exactly four deletions.
Furthermore, observe that~$$\mathcal{Q}:=\{\{t^1_j,p^9_j,p^5_j,p^7_j\}, \{t^2_j,p^{10}_j,p^6_j,p^8_j\},\{p^7_j,p^3_j,p^1_j,r^1_j\},\{p^8_j,p^4_j,p^2_j,p^1_j\}\}$$ is a packing of four edge-disjoint bicolored~$P_4$s.
Also, observe that~$\{p^5_j,p^3_j,p^1_j,p^2_j\}$ is also a bicolored~$P_4$ in~$G$ which only has the edge~$\{p^1_j,p^3_j\}$ in common with each bicolored~$P_4$ in~$\mathcal{Q}$.
Thus,~$\{p^1_j,p^3_j\}\in S$.
Analogously, we can show that~$\{p^2_j,p^4_j\}\in S$.

Furthermore, observe that~$\{t^1_j,p^9_j,p^5_j,p^3_j\}$ in also a bicolored~$P_4$ in~$G$.
Hence, $\{t^1_j,p^9_j\}\in S$ or $\{p^9_j,p^5_j\}\in S$. 
Assume that~$\{p^9_j,p^5_j\}\in S$.
Since~$t^1_j$ is the only vertex in that~$P_4$ which has neighbors outside~$X_j$, the set~$S':=S-\{p^9_j,p^5_j\}\cup\{\{t^1_j,p^9_j\}\}$ is also an edge-deletion set with exactly nine deletions in~$X_j$.
Analogously, we can show that also~$\{t^2_j,p^{10}_j\}\in S'$.
 $\hfill \Diamond$
\end{claimproof}

Now, we are ready to prove the correctness.

$(\Rightarrow)$ Let~$\mathcal{A}:\mathcal{X}\rightarrow\{\texttt{true},\texttt{false}\}$ be a satisfying assignment for~$\Phi$. We will construct an edge-deletion set~$S$ of size~$k$ such that~$G- S$ is bicolored $P_4$-free.

For each variable~$x_j$ we do the following: 
If~$\mathcal{A}(x_j)=\texttt{true}$, then we delete the first edge set~$\mathcal{E}_1$ described in Claim~\ref{Claim-variable-gadget-2-col-p4-del}.
Otherwise, if~$\mathcal{A}(x_j)=\texttt{false}$, then we delete the second edge set~$\mathcal{E}_2$ described in Claim~\ref{Claim-variable-gadget-2-col-p4-del}.
Since~$\mathcal{A}$ satisfies~$\Phi$, there is at least one variable~$x_j$ such that~$\mathcal{A}(x_j)$ satisfies clause~$c_i$. 
Let~$z\in[3]$ be the~$z$-th literal of~$c_i$ that satisfies clause~$c_i$. 
Furthermore, let~$\{\alpha,\beta\}:=[3]\setminus\{z\}$. 
We delete the red edges~$\{c_i^\alpha,d_i^\alpha\},$~and~$\{c_i^\beta,d_i^\beta\}$. 
Observe that we deleted exactly nine edges per variable gadget and exactly two edges per clause gadget.
Hence,~$|S|=k=9\eta +2\mu$. 

It remains to show that~$G-S$ is bicolored~$P_4$-free. 
From Claims~\ref{Claim-clause-gadget-2-col-p4-del} and~\ref{Claim-variable-gadget-2-col-p4-del} we conclude that there is no bicolored~$P_4$ whose vertex set is entirely contained in one clause or variable gadget.
Recall that the vertices~$c_i^z$ for~$z\in[3]$ which are identified with the vertices~$t_j^y,$ and~$f_j^y$ for~$y\in[2]$ are the only vertices which connect variable and clause gadgets in~$G$. 
Observe that all of these vertices have degree two.
Consider a fixed vertex~$c_i^z$.
Next, we prove that in~$G-S$ vertex~$c_i^z$ has degree at most one. 
Without loss of generality we assume that vertex~$c_i^z$ is identified with vertex~$t_j^y$ from variable~$x_j$. 
If clause~$c_i$ is not satisfied by its~$z$-th literal, then~$\{c_i^z,d_i^z\}\in S$, and vertex~$c_i^z$ has degree at most one in~$G-S$.
Hence, in the following we assume that clause~$c_i$ is satisfied by its~$z$-th literal which is~$x_j$. 
Thus~$\mathcal{A}(x_j)=\texttt{true}$ and according to our definition of~$S$ we observe that~$\{t_j^y,p_j^{8+y}\}=\{c_i^z,p_j^{8+y}\}\in S$. 
We conclude that each vertex~$c_i^z$ in~$G-S$ has degree at most one. 
Thus, there is no bicolored~$P_4$ in at least two gadgets. 
Hence,~$G-S$ is bicolored~$P_4$-free.

$(\Leftarrow)$ Let~$S\subseteq E(G)$ with~$|S|\le k$ be an edge set such that~$G-S$ is bicolored~$P_4$-free. 
In the following we will construct a satisfying assignment~$\mathcal{A}:\mathcal{X}\to\{\texttt{true, false}\}$ for~$\Phi$.

Since~$|S|=9\eta +2\mu$,~$|S\cap E(Z_i)|\ge 2$ for each clause gadget~$Z_i$ by Claim~\ref{Claim-clause-gadget-2-col-p4-del}, and~$|S\cap E(X_j)|\ge 9$ for each variable gadget~$X_j$ by Claim~\ref{Claim-variable-gadget-2-col-p4-del}, we conclude that~$|S\cap E_r(Z_i)|=2$ and that~$|S\cap E(X_j)|= 9$.
Next, we construct an edge-deletion set~$S'$ with~$|S'|=|S|=k$ which fulfills the conditions of Claim~\ref{Claim-variable-gadget-2-col-p4-del}, that is for each variable~$X_j$ the set~$S'$ contains either the edge set~$\mathcal{E}_1$ or the edge set~$\mathcal{E}_2$ described in Claim~\ref{Claim-variable-gadget-2-col-p4-del}.
If~$S'$ contains the first edge set~$\mathcal{E}_1$ described in Claim~\ref{Claim-variable-gadget-2-col-p4-del}, we set~$\mathcal{A}(x_j)=\texttt{true}$ and otherwise, if~$S'$ contains the second edge set~$\mathcal{E}_2$ described in Claim~\ref{Claim-variable-gadget-2-col-p4-del}, we set~$\mathcal{A}(x_j)=\texttt{false}$.

It remains to show that~$\mathcal{A}$ satisfies~$\Phi$.
Consider the clause gadget~$Z_i$. 
By Claim~\ref{Claim-clause-gadget-2-col-p4-del} we conclude that~$|S'\cap E_r(Z_i)|=2$. 
Hence, there exist a~$z\in[3]$ such that~$\{c_i^z,d_i^z\}\notin S$.
Without loss of generality we assume that vertex~$c_i^z$ is identified with vertex~$t_j^y$ for some variable~$x_j$ and some~$y\in[2]$. Observe that~$G[\{d_i^\alpha,d_i^z,c_i^z,p_j^{8+y}\}]$ for some~$\alpha\in[3]\setminus\{z\}$ is an induced bicolored~$P_4$ in~$G$. 
Since~$G-S$ is bicolored~$P_4$-free, we conclude that~$\{c_i^z,p_j^{8+y}\}\in S$.
Thus,~$\mathcal{A}(x_j)=\texttt{true}$ and hence clause~$c_i$ is satisfied by~$\mathcal{A}$.
We conclude that~$\Phi$ is satisfied by~$\mathcal{A}$. 
\end{proof}

Before we present the polynomial-time algorithm for~$2P_4$D when the input graph is from~$\mathcal{T}$, we characterize~$\mathcal{T}$. More precisely, we show that each connected component in some~$G \in \mathcal{T}$ is one of the following two graphs (see Fig.~\ref{Figure: fence and clique-star}): An \emph{rb-fence} is a graph that consists of exactly two blue cliques~$X$ and~$Y$ of size at least two and the red edges form a matching between vertices of~$X$ and~$Y$. An \emph{rb-clique-star} is a graph that consists of exactly one red clique~$X$ and up to~$|X|$ non-overlapping blue cliques where each blue clique intersects with~$X$ in a unique vertex. To show this characterization of~$\mathcal{T}$, we first prove the following.

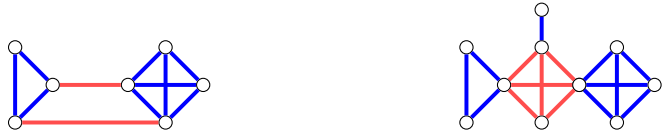
\begin{figure}[t]
\begin{center}
\begin{tikzpicture}
\tikzstyle{knoten}=[circle,fill=white,draw=black,minimum size=5pt,inner sep=0pt]

\tikzstyle{bez}=[inner sep=0pt]

\node[knoten] (u)  at (0,0) {};
\node[knoten] (v)  at (0,1) {};
\node[knoten] (w)  at (0.5,0.5) {};

\node[knoten] (x)  at (2,0) {};
\node[knoten] (y)  at (2,1) {};
\node[knoten] (z)  at (1.5,0.5) {};
\node[knoten] (q)  at (2.5,0.5) {};

\draw[blue edge]  (u) to (v);
\draw[blue edge]  (u) to (w);
\draw[blue edge]  (w) to (v);
\draw[blue edge]  (x) to (y);
\draw[blue edge]  (z) to (y);
\draw[blue edge]  (z) to (x);
\draw[blue edge]  (q) to (x);
\draw[blue edge]  (z) to (q);
\draw[blue edge]  (y) to (q);

\draw[red edge]  (w) to (z);
\draw[red edge]  (u) to (x);

\begin{scope}[xshift=7cm]
\node[knoten] (x)  at (0,0) {};
\node[knoten] (y)  at (-0.5,0.5) {};
\node[knoten] (z)  at (0.5,0.5) {};
\node[knoten] (q)  at (0,1) {};

\node[knoten] (a) at (-1, 1) {};
\node[knoten] (b) at (-1, 0) {};

\node[knoten] (c) at (0,1.5) {};

\node[knoten] (s)  at (1,0) {};
\node[knoten] (t)  at (0.5,0.5) {};
\node[knoten] (u)  at (1.5,0.5) {};
\node[knoten] (v)  at (1,1) {};

\draw[red edge]  (x) to (y);
\draw[red edge]  (z) to (y);
\draw[red edge]  (z) to (x);
\draw[red edge]  (q) to (x);
\draw[red edge]  (z) to (q);
\draw[red edge]  (y) to (q);

\draw[blue edge]  (a) to (y);
\draw[blue edge]  (b) to (y);
\draw[blue edge]  (a) to (b);

\draw[blue edge]  (q) to (c);

\draw[blue edge]  (s) to (t);
\draw[blue edge]  (s) to (u);
\draw[blue edge]  (s) to (v);
\draw[blue edge]  (u) to (v);
\draw[blue edge]  (u) to (t);
\draw[blue edge]  (v) to (t);
\end{scope}

\end{tikzpicture}
\caption{Left: An example of an rb-fence. Right: An example of an rb-clique-star.}\label{Figure: fence and clique-star} 
\end{center} 
\end{figure}

\begin{lemma}
\label{claim-charactierzation-bicolored-p4-free-cluster}
Let~$\{v,w\}$ be a blue edge in a bicolored graph~$G \in \mathcal{T}$ such that~$v$ has a red neighbor~$u$ and~$w$ has a red neighbor~$x$. 
Then,~$u  \neq x$ and~$\{u,x\}\in E(G)$ with color blue. 
Furthermore, neither~$u$ and~$v$ nor~$w$ and~$x$ have a common red neighbor.
\end{lemma}

\begin{proof}
Observe that~$u$ and~$x$ are distinct vertices, since otherwise~$G[\{u,v,w\}]$ is a bicolored triangle, a contradiction to~$G \in \mathcal{T}$. 
Since the vertices~$u,v,w$, and~$x$ form a bicolored~$P_4$ with subsequent edge colors red, blue and red, we conclude that this bicolored~$P_4$ is not induced. 
Observe that~$\{u,w\}\notin E(G)$ and that~$\{v,x\}\notin E(G)$ since otherwise~$G$ would contain a bicolored triangle.
Hence,~$\{u,x\}\in E(G)$.
Furthermore,~$\{u,x\}$ is blue, since otherwise the red subgraph~$G_r$ is no cluster graph.
Next, assume towards a contradiction that at least one of both red edges has a common red neighbor.
Without loss of generality, we assume that vertices~$w$ and~$x$ have a common red neighbor~$y$.
Observe that~$\{v,y\}\notin E(G)$ since otherwise~$G$ would contain a bicolored triangle. 
Next, consider the case that~$\{u,y\}\in E(G)$.
With symmetric arguments we conclude~$\{v,y\}\notin E(G)$.
Hence,~$\{u,y\}$ is blue. 
But now~$G[\{v,u,y,x\}]$ is an induced bicolored~$P_4$ with subsequent edge colors red, blue, red; a contradiction.
\end{proof}

We now present our alternative characterization of~$\mathcal{T}$.

\begin{proposition}
\label{lemma-alternative-characterization}
Let~$G$ be a graph. Then,~$G \in \mathcal{T}$ if and only if each connected component of~$G$ is either an rb-fence or an rb-clique-star.
\end{proposition}

\begin{proof}
Given an rb-fence or an rb-clique-star, the colored subgraphs of both colors are cluster graphs. Moreover, it is easy to see that no such graph has an induced bicolored~$P_4$ with subsequent edge colors~red, blue, red. Thus, every graph where each connected component is either an rb-fence or an rb-clique-star belongs to~$\mathcal{T}$. We next show that each graph in~$\mathcal{T}$ is a disjoint union of rb-clique-stars and rb-fences.

First, consider a connected component~$Z$ of~$G$ which contains a maximal red clique~$C$ with at least three vertices. 
Observe that each vertex of~$C$ does not have a red neighbor outside~$C$. Thus, each~$v \in C$ only blue neighbors outside~$C$. Since the blue graph is a cluster graph, these blue neighbors of~$v$ form a clique. Moreover, no two vertices in~$C$ have a common blue neighbor since otherwise~$G$ has a bicolored triangle. Finally, each blue neighbor of~$C$ has no red neighbor outside~$C$ since otherwise~$G$ has a red-blue-red~$P_4$. Thus, the connected component~$Z$ is an rb-clique-star.

Second, consider a connected component~$Z$ of~$G$ which does not contain red triangles.
In other words, the red connected components of~$Z$ are edges.
Let~$\{u,v\}$ and~$\{x,y\}$ be two red edges in~$Z$. 
In the next step, we will prove that~$G[\{u,v,x,y\}]$ is a~$C_4$ where two subsequent edges have different colors.
Assume towards a contradiction that this is not the case. 
By Lemma~\ref{claim-charactierzation-bicolored-p4-free-cluster} we conclude that if~$G[\{u,v,x,y\}]$ contains one blue edge, then~$G[\{u,v,x,y\}]$ has to be a~$C_4$
where two subsequent edges have different colors. 
Hence, we assume that~$G[\{u,v,x,y\}]$ does not contain any blue edges.
Without loss of generality assume that ~$P:=(u=a_0,a_1,\ldots,a_p=x)$ is a shortest path connecting the edges~$\{u,v\}$ and~$\{x,y\}$.
Observe that since the induced red and the induced blue subgraph are both cluster graphs, the edges in~$P$ have alternating colors. 
Hence,~$\{u,a_1\}$ is blue and~$\{a_1,a_2\}$ is red. 
Furthermore, note that~$a_1\neq x$ and~$a_1\neq y$ since otherwise there is a blue edge in~$G[\{u,v,x,y\}]$.
Observe that~$G[\{u,v,a_1,a_2\}]$ contains a bicolored~$P_4$ with subsequent edge colors red, blue, red. 
By Lemma~\ref{claim-charactierzation-bicolored-p4-free-cluster} we conclude that~$\{v,a_2\}\in E(G)$ with color blue. 
Then,~$P':=\{v=a_1,\ldots,a_p=x\}$ is a shorter path connecting the edges~$\{u,v\}$ and~$\{x,y\}$ than~$P$, a contradiction.
We conclude that for each two red edges~$\{u,v\}$ and~$\{x,y\}$ in~$Z$ the induced subgraph~$G[\{u,v,x,y\}]$ is a~$C_4$ where two subsequent edges have different colors.

Let~$R:=\{\{u_i,v_i\} \mid i\in[t]\}$ be the set of red edges in~$Z$. Since for every pair~$\{u_i,v_i\}, \{u_j,v_j\}$ with~$i \neq j$ the graph~$G[\{u_i,v_i,u_j,v_j\}]$ is a~$C_4$ we may assume without loss of generality that~$B_1:= \{u_i \mid i \in [t]\}$ and~$B_2:= \{v_i \mid i \in [t]\}$ are blue cliques in~$G$. Observe that~$E_G(B_1,B_2)=R$ since otherwise, if there is an edge~$\{u_i,v_j\}$ with~$i \neq j$, then~$G[\{u_i,u_j,v_j\}]$ or~$G[\{v_i,u_j,v_j\}]$ is a bicolored triangle in~$G$.

If no vertex in~$B_1 \cup B_2$ has a neighbor outside~$B_1 \cup B_2$, then~$Z=B_1 \cup B_2$ and therefore,~$G[Z]$ is an rb-fence. So, without loss of generality, let~$x$ be a neighbor of~$u_1$ outside~$B_1 \cup B_2$. 
Then,~$\{u_i,x\}$ is blue and since the blue graph is a cluster graph it follows that~$B_1 \cup \{x\}$ is a blue clique in~$G$. 
This implies that~$B_1 \subseteq K_1$ and~$B_2 \subseteq K_2$ for some maximal blue cliques~$K_1$ and~$K_2$ and that~$G[Z]=G[K_1 \cup K_2]$  is an rb-fence.
\end{proof}

The next goal is to show that~$2P_4$D can be solved in polynomial time on graphs in~$\mathcal{T}$. To this end, we first consider the solution structure of rb-fences.

\begin{lemma} \label{Lemma: Solution on Fences}
Let~$(G,k)$ be an instance of $2P_4$D where~$G$ is an rb-fence and~$S$ be a minimum size solution. Then either~$G$ is a~$C_4$ and~$S= \emptyset$ or~$S$ contains exactly the red edges of~$G$.
\end{lemma}

\begin{proof}
Let~$K_1$ and~$K_2$ with~$|K_1|\ge |K_2|\ge 2$ be the two blue cliques of~$G$, and let~$B_1\subseteq K_1$ and~$B_2\subseteq K_2$ be such that there is a perfect red matching between~$B_1$ and~$B_2$, and let~$|B_1|$ and~$|B_2|$ be maximal under this property.
Furthermore, by~$R$ we denote the set of red edges in~$G$. 
If~$|K_1|=2$ and~$|R|=2$, then~$G$ is a bicolored~$C_4$ with alternating colors.
Hence,~$G$ is~$2P_4$-free.
Next, we assume that this is not the case. 
In the following, we prove that there is a minimal solution that deletes all red edges. Obviously,~$G-R$ is~$2P_4$-free. We next show that there is no solution of size smaller than~$|R|$.

First, consider the case~$|K_2|\ge 3$ or~$|R|=1$. 
Then, there exists a vertex-pair disjoint packing of bicolored~$P_4$s of size~$|R|$: Each red edge is the central edge of one of the~$P_4$s of this packing. Since one edge-deletion transforms only one bicolored~$P_4$ of the packing, there have to be at least~$|R|$ edge deletions in a minimal solution.

Second, consider the case~$|K_2|=2=|R|$. 
Recall that~$|K_1|\ge 3$. 
Since~$G$ contains a bicolored~$P_4$ and deleting any of its three edges does not make~$G$ bicolored~$P_4$-free, we conclude that at least two edges of~$G$ have to be deleted.
Hence, deleting both red edges is optimal. 
\end{proof}

Lemma~\ref{Lemma: Solution on Fences} implies that connected components that are rb-fences can be solved in linear time. We next study the solution structure of rb-clique-stars.

\begin{lemma} \label{Lemma: Solution of clique-stars}
Let~$(G,k)$ be an instance of~$2P_4D$ where~$G$ is an rb-clique-star. Let~$C$ be the red clique and let~$B_1, \dots, B_\ell$ be the blue cliques of~$G$ such that~$|B_1| \geq \dots \geq |B_\ell|$. Furthermore, for every~$i \in [\ell]$, let~$c_i$ denote the unique vertex in~$B_i \cap C$.

Then, there is a solution~$S$ such that~$G-S$ consists of~$\ell$ connected components and there is some~$p \in [\ell]$ with
\begin{enumerate}
\item[a)] $B_1 \cup C \setminus \{c_q \mid q \in [2,p]\}$ is a connected component in~$G-S$,
\item[b)] for every~$q \in [2,p]$, $B_q$ is a connected component in~$G-S$, and
\item[c)] for every~$q \in [p+1,\ell]$, $B_q \setminus \{c_q\}$ is a connected component in~$G-S$.
\end{enumerate}
\end{lemma}

\begin{proof} 
Let~$S$ be an edge-deletion set of size at most~$k$ of~$G$ and let~$Z_1,\ldots, Z_p$ be the connected components of~$G-S$.

We prove the lemma in three steps.
 
\textbf{Step 1:} We first show that we may assume that for each~$j \in [\ell]$ the vertex set~$B_j \setminus \{c_j\}$ is completely contained in one connected component of~$G-S$.

Observe that $B_j \setminus \{c_j\}$ is a colored neighborhood class by Definition~\ref{Def:ColNeighDiv}. Moreover,~$N(B_j \setminus \{c_j\})=\{c_j\}$. Then, by Proposition~\ref{Proposition: Equal Choice Neighborhood Class} we may assume that no edge in~$B_j \setminus \{c_j\}$ is part of a minimal solution~$S$ and that either~$E(\{c_j\}, B_j \setminus \{c_j\}) \subseteq S$ or~$E(\{c_j\}, B_j \setminus \{c_j\}) \cap S= \emptyset$. Consequently, we can safely assume that for each~$j\in[\ell]$ the vertices in~$B_j\setminus \{c_j\}$ are in one connected component~$Z_i$ for some~$i\in[p]$.

\textbf{Step 2:} We next show that we may assume that the connected component containing~$B_1 \setminus \{c_1\}$ is the only connected component in~$G-S$ that might contain red edges.

Assume that there exist two connected components~$Z_i$ and~$Z_j$ which contain red edges. 
By~$R_i$ we denote the vertices in~$Z_i$ which are not incident with blue edges.
We define~$P:=Z_j\cup R_i$ and~$Q:=Z_i\setminus R_i$. Then,~$G[P]$ and~$G[Q]$ are also~$P_4$-free. It remains to show that a solution introducing the connected components~$P$ and~$Q$ instead of~$Z_j$ and~$Z_i$ is optimal. Note that there are at least~$(|R_i|+1) \cdot (|R_j|+1)$ edge deletions in~$E(Z_i \cup Z_j)$ to obtain connected components~$Z_i$ and~$Z_j$. In contrast, there are only~$|R_i|$ edge deletions in~$E(Z_i \cup Z_j)$ to obtain connected components~$P$ and~$Q$. Since the number of edge deletions between~$Z_i \cup Z_j$ and the rest of the graph is the same in both cases we conclude that this modification of the solution does not increase the number of edge deletions.
Thus, in the following we can assume that there is at most one connected component which contains red edges.

We next prove that we can safely assume that if a connected component of~$G-S$ contains red edges, it also contains~$B_1$.
Assume that this is not the case.
So let~$Z_1$ be the connected component containing~$B_1\setminus\{c_1\}$ and let~$Z_j$ be the connected component containing all red edges.
Let~$c_q\in V(Z_j)$ for some~$q\in[2,\ell]$ be the vertex incident with red and blue edges. 
If such a vertex~$c_q$ does not exist, then~$Z_j$ is a red clique and thus~$G[Z_i\cup Z_j]$ is also bicolored~$P_4$-free and the statement follows.
Otherwise, vertex~$c_j$ exists. 
Let~$Y:=Z_j\setminus B_j$ the vertices of~$Z_j$ which are only incident with red edges.
Observe that~$Z_j\setminus Y$ and~$Z_i\cup Y$ are also bicolored~$P_4$ free and the costs of this solution are also at most~$k$ since~$|B_1|\ge |B_j|$. 
Hence, in the following we assume that the connected component of~$G-S$ containing the red edges, also contains~$B_1$.

\textbf{Step 3:} We finally show that there is some~$p \leq \ell$ satisfying statements~$a)$--$c)$. First, observe that each connected component contains exactly one blue clique~$B_i \setminus \{c_i\}$, since otherwise, the component contains a bicolored~$P_4$. Let~$Z_i$ for~$i\in [\ell]$ be the connected component that contains~$B_i \setminus \{c_i\}$. 

Assume that there are some~$i$ and~$j$ with~$1<i<j$ such that~$Z_i = B_i \setminus \{c_j\}$ and~$Z_j = B_j$. Note that~$i<j$ implies~$|B_i| \geq |B_j|$. We then define the sets~$Z_1':= Z_1 \setminus \{c_i\} \cup \{c_j\}$, $Z_i':= B_i$, and~$Z_j':= B_j \setminus \{c_j\}$. It is easy to see that~$G[Z_1']$,~$G[Z_i']$, and~$G[Z_j']$ contain no induced bicolored~$P_4$. Moreover, note that, if~$a$ is the number of edges in the disjoint union of~$G[Z_1']$,~$G[Z_i']$, and~$G[Z_j']$ and~$b$ is the number of edges in the disjoint union of~$G[Z_1]$,~$G[Z_i]$, and~$G[Z_j]$, we have~$a-b = |B_i|-|B_j|$. Together with the fact that~$|B_i| \geq |B_j|$ we conclude that there is an optimal solution introducing the connected components~$Z_1'$,~$Z_i'$, and~$Z_j'$ instead of~$Z_1,Z_i$, and~$Z_j$.

Consequently, we may assume that there is some~$p \in [\ell]$ that satisfies statements~$b)$ and~$c)$. The component~$Z_1$ then contains all vertices that are not in one of the components~$B_q$ with~$q \in [2,p]$ or~$B_q \setminus \{c_q\}$ with~$q \in [p+1, \ell]$. Thus,~$Z_1 = B_1 \cup C \setminus \{c_q \mid q \in [2,p]\}$ and therefore, statement~$a)$ holds. 
\end{proof}

Lemma~\ref{Lemma: Solution of clique-stars} implies that, in a~$2P_4$D-instance~$(G,k)$, the connected components of~$G$ that are rb-clique-stars can be solved in polynomial time:
Let the rb-clique-star contain~$\ell$ blue cliques. Then, for every~$p \leq \ell$, compute in~$\Oh(nm)$~time the cost of an edge-deletion set~$S$ that forms a partition according Lemma~\ref{Lemma: Solution of clique-stars} and keep the set with minimal cost. Altogether, this implies the following.

\begin{theorem}
$2P_4$D can be solved in ~$\mathcal{O}(nm)$~time on graphs in~$\mathcal{T}$.
\end{theorem}

\section{Conclusion}

We initiated a study of edge-deletion problems where the aim is to destroy paths or cycles containing a certain number of colors. We left many problems open for future research. 

First, in our analysis of the classic complexity of \textsc{$cP_{\ell}$D} for~$\ell \geq 4$ on bounded-degree graphs, we have shown that for~$c \in [2,\ell-2]$ we obtain NP-hardness even if the maximum degree is~$3$, while for~$c=\ell-1$, we obtain hardness even if the maximum is~$16$.  It is thus a natural question whether the NP-hardness for~$c=\ell-1$ also holds for a smaller maximum degree.

Furthermore, it is open whether there is a polynomial-time algorithm for \textsc{$(\ell-1)P_{\ell}$D} on non-cascading graphs for any~${\ell\geq4}$ since we only showed the NP-hardness on non-cascading graphs for each~$\ell\geq 4$ and each $c\in[\ell-2]$. 
Similarly, we showed W[2]-hardness of CPD, the variant of $cP_\ell$D where~$\ell$ is part of the input and~$k$ is the parameter for non-cascading graphs, but it remains open whether CPD is W[2]-hard for $k$ when every $c$-colored $P_\ell$ in the input graph is an induced $c$-colored~$P_\ell$. 
Concerning the parameterized complexity with respect to the colored neighborhood diversity~$\gamma$ it remains open whether $2C_4$D is fixed-parameter tractable with respect to~$\gamma$. It might also be interesting to analyze the parameterized complexity with regard to further structural graph parameters like the vertex cover number or the size of a minimum feedback vertex set. Moreover, since~$\gamma$ can be seen as a colorful version of the classic parameter \emph{neighborhood diversity}, it can be interesting to investigate whether further well-known structural parameters have a corresponding colorful version that is useful for the problems studied in this work.
It would also be interesting to study the corresponding edge-completion problems and edge-modification problems from an algorithmic point of view. Also, to which extent is it possible to adapt our results to related problems on vertex-colored graphs instead of edge-colored graphs?

More generally, it seems interesting to study further whether the fact that an input graph
for an edge-modification problem is non-cascading has any impact on the problem
difficulty. In other words, which edge-modification problems that are generally NP-hard
become polynomial-time solvable on non-cascading input graphs? This question is also
relevant for problems in uncolored graphs. Finally, developing a deeper understanding of
properties of edge-colored graphs appears to be a wide open and fruitful research
topic. Our work makes two contributions in this direction: characterizing one graph class
via forbidden induced subgraphs and extending the neighborhood diversity parameter to
edge-colored graphs. What are further interesting classes of edge-colored graphs and further graph parameters that take the edge-coloring into account? 

\acknowledgements
We would like to thank the reviewers of \emph{Discrete Mathematics and Theoretical Computer Science} for their helpful comments. Some of the results of this work are also contained in the first author's Bachelor thesis~\cite{NJE20}.

\bibliographystyle{alpha}
\bibliography{ref}



\end{document}